\documentclass[11pt]{article}

\usepackage[a4paper,margin=1in]{geometry}
\usepackage{authblk}

\usepackage[T1]{fontenc}
\usepackage{amsmath}
\usepackage{amssymb}
\usepackage{amsfonts}
\usepackage{amsfonts}
\newlength{\mathindent}
\usepackage{amsthm}
\usepackage{bm}

\newtheorem{theorem}{Theorem}[section]
\newtheorem{lemma}[theorem]{Lemma}
\newtheorem{proposition}[theorem]{Proposition}

\theoremstyle{definition}
\newtheorem{definition}[theorem]{Definition}
\newtheorem{example}[theorem]{Example}
\theoremstyle{remark}

\usepackage{graphicx}
\usepackage{booktabs}
\usepackage[compatibility=false]{caption}
\usepackage[matrix,arrow,cmtip,curve]{xy}
\usepackage{hyperref}
\usepackage{tikz,tikz-3dplot,tikz-cd}
\usepackage{caption}
\usepackage{subcaption}
\usepackage{paralist}
\usepackage{comment}

\usepackage{tabularx}
\usepackage{diagbox}
\usepackage{amsfonts}
\usepackage{graphicx}
\usepackage{multicol}
\usepackage[matrix,arrow,cmtip,curve]{xy}
\usepackage{amsmath}
\usepackage{tikz,tikz-3dplot,tikz-cd}
\usepackage{caption}
\usepackage{subcaption}
\usepackage{paralist}
\usepackage{comment}
\DeclareMathAlphabet{\mathbbm}{U}{bbm}{m}{n}          % blackboard bold
\newcommand*\twotwo{\textsf{2+2}}

\newcommand*\ev{\textup{\textsf{ev}}}

\newcommand*{\id}{\mathrm{id}}

\newcommand*{\Path}{\mathrm{Path}}
\newcommand*{\cell}{\mathrm{Cell}}
\newcommand*{\SG}{\mathfrak{S}} %symmetric group

\newcommand*\exec{%
  \raisebox{1pt}{%
    \begin{tikzpicture}[x=.8ex,y=1ex,-]
      \draw (0,0) -- (1,0) -- (1,1) -- (2,1);
    \end{tikzpicture}}}

\newcommand*\ibullet{\vcenter{\hbox{\tiny $\bullet$}}}

\newcommand*\squarefull{\Box_{full}}% square full
\newcommand*\Xifull{\Xi_{full}}
\newcommand*\si{\mathrm{s_P}} %source interface map
\newcommand*\ti{\mathrm{t_P}} %target interface map
\newcommand*\sev{\textsf{ev}}

\makeatother
\newcommand\pomsetwop[4]{% @R; @C; @M; content
  \vcenter{\xymatrix@1@R=#1@C=#2@M=#3{#4}}%
}
\newcommand\pomset[2][1.3]{%
  \left[%
    \pomsetwop{0ex}{#1em}{2pt}{#2}%
  \right]%
}
\newcommand\ipomset[2][1.5]{%
  \hspace*{.6em}%
  \left(%
    \hspace*{-#1em}%
    \pomsetwop{2.7ex}{#1em}{1pt}{#2}%
    \hspace*{-#1em}%
  \right)%
  \hspace*{.6em}%
}
\newcommand\bigipomset[2][1.3]{%
  \hspace*{.1em}%
  \left[%
    \hspace*{-#1em}%
    \pomsetwop{2.5ex}{#1em}{1pt}{#2}%
    \hspace*{-#1em}%
  \right]%
  \hspace*{.1em}%
}

\usetikzlibrary{calc,decorations.pathreplacing,positioning}
\RequirePackage[normalem]{ulem}

\tikzstyle{io} = [trapezium, 
trapezium stretches=true, % A later addition
trapezium left angle=70, 
trapezium right angle=110, 
minimum width=3cm, 
minimum height=1cm, text centered, 
draw=black, fill=blue!30]

\tikzstyle{process} = [rectangle, 
minimum width=3cm, 
minimum height=1cm, 
text centered, 
text width=3cm, 
draw=black, 
fill=orange!30]

\tikzstyle{decision} = [diamond, 
minimum width=3cm, 
minimum height=1cm, 
text centered, 
draw=black, 
fill=green!30]
\tikzstyle{arrow} = [thick,->,>=stealth]
\tikzset{->, auto, >=stealth', font=\small}
\tikzset{state/.style={shape=circle, draw, fill=white, initial text=,
    inner sep=.5mm, minimum size=1.5mm}}

\tikzset{accepting/.style=accepting by arrow}
\tikzset{state with output/.style={shape=rectangle split, rectangle
    split parts=2, draw, fill=white,
    initial text=, inner sep=1mm}}
    \usetikzlibrary{arrows.meta, positioning} 
\usetikzlibrary{fit, arrows, automata, shapes, calc, fadings,shapes.geometric}
\tikzstyle{startstop} = [rectangle, rounded corners, 
minimum width=3cm, 
minimum height=1cm,
text centered, 
draw=black, 
fill=red!30]
\begin{document}
\title{Forgetting Event Order in Higher-Dimensional Automata}

\author[1]{Safa Zouari}
\affil[1]{Norwegian University of Science and Technology, Gj{\o}vik, Norway\\\texttt{safa.zouari@ntnu.no}}
\date{}

\maketitle

\begin{abstract}
 Higher-dimensional automata (HDAs) provide a geometric model of true concurrency, yet their standard formulation encodes an artificial total order on events. This representational artifact causes a fundamental mismatch between the combinatorial structure of HDAs and their observable behavior, leading to logical asymmetries and complicating the application of categorical tools. In this paper, we resolve this tension by developing a semantics for HDAs that is independent of event order, based on interval ipomsets (partially ordered multisets with interfaces) that preserve only precedence and concurrency. We prove that for any HDA, the traditional ST–trace of an execution path corresponds precisely to its associated interval ipomset. On the structural side, we show that the presheaf-theoretic presentation with an unordered base and the combinatorial presentation of symmetric HDAs are categorically isomorphic. Finally, by characterizing ST- and hereditary history-preserving (hhp) bisimulation via ipomset isomorphism, we provide a unified, order-free foundation for HDA semantics. Our results resolve several critical ambiguities in the literature: they provide the necessary path-category structure to canonically apply the Open Maps framework, eliminate representational artifacts in temporal and modal logics, and bridge systematic mismatches between HDAs and other models of concurrency such as Petri nets.
\end{abstract}

\section{Introduction}
Higher-dimensional automata (HDAs) offer a geometric view of true concurrency: higher-dimensional cells represent independent events occurring simultaneously rather than through artificial interleaving. This structure makes HDAs a unifying semantic framework for concurrency, as Petri nets, event structures, and asynchronous transition systems all embed into HDAs up to hereditary history-preserving bisimulation~\cite{PrattCG,VANGLABBEEK2006265,nielsen1981petri,van1995configuration,bednarczyk1989categories,shields1985concurrent}. 

A central question in HDA theory is the abstract representation of execution behavior. Two formalisms traditionally dominate: ST-traces, which record the causal evolution of events via start/termination observations~\cite{VANGLABBEEK2006265,vG91}, and ipomset labels—partially ordered multisets with interfaces—which provide a compositional language theory~\cite{LanguageofHDA,KleeneTh,MyhillNerode,amrane2023developments}. Although both aim to capture the same observable behavior, their relationship remains subtle and, at times, contradictory. In existing work, pomset-based semantics impose an extraneous event order inherited from the HDA's combinatorial structure. This order lacks a semantic counterpart in models like Petri nets and is absent from ST-trace observations. Consequently, pomset labels frequently distinguish executions that are observationally equivalent, creating a conceptual tension in behavioral reasoning. This mismatch has tangible consequences for logic. In temporal and modal logics for HDAs~\cite{amrane2025buchielgottrakhtenbrottheoremhigherdimensionalautomata,Clement2024KampTF,logiclang,Safa}, the imposed order breaks symmetry: a formula may accept \(a\parallel b\) but reject \(b\parallel a\),
even though these executions are observationally indistinguishable.

Such asymmetries are artifacts of the representation rather than of the underlying logic. Moreover, injecting Petri-net behaviors into ordered HDA frameworks leads to systematic mismatches \cite{efficintconversion, petriandhda}, pointing to a deeper structural problem: the standard event ordered representation is fundamentally incompatible with the categorical tools used to reason about concurrency. Specifically, while the Open Maps framework \cite{JOYAL1996164} provides a standard recipe for deriving modal logics that characterize \emph{hereditary history preserving bisimulation} (hhp-bisimulation), its application to HDAs has remained ambiguous. The reliance on ST-traces lacks the necessary path-category structure, leaving the relationship between observational traces and categorical logic poorly defined.

The event order carried by the standard (ordered) presentation of HDAs is a representational artifact: it is not part of the observable execution content captured by trace-based semantics. To provide a robust foundation that is applicable to all HDAs, we leverage a fundamental result from Kahl's combinatorial theory, which proves that every HDA is hhp-bisimilar to its symmetric expansion \cite{KAHL202247}. We formally justify this approach by proving that our adopted presheaf-theoretic framework of \cite{struth2024presheaf} is categorically isomorphic to Kahl's symmetric framework, ensuring that our results are universally applicable to the entire class of all HDAs, not just the symmetric ones.

We therefore move beyond the standard ordered presentation of HDAs \cite{LanguageofHDA,KleeneTh}. By adopting Fishburn's original formulation of pomsets—defined as labeled partial orders— \cite{fishburn1970intransitive} and employing gluing composition \cite{generatingposetbeyondn}, we show that the observable content of an HDA path is precisely an
\emph{interval pomset with interfaces}, retaining only precedence and canonical interfaces.
This makes it possible to identify executions via ipomset isomorphism and yields an order-free semantics
that agrees with ST-observations and restores symmetry
(e.g.\ $(a \parallel b) \cong (b \parallel a)$). Our results are:
\begin{itemize}
    \item \textbf{Forgetting the Event Order in HDA Structure (Section \ref{sec:base cat} \& \ref{sec:HDA}):} We prove that the category of HDAs over an order-free base—following the presheaf approach of Struth et al. \cite{struth2024presheaf}—is isomorphic to the category of symmetric HDAs Kahl \cite{KAHL202247}. This result is crucial as it imports the Symmetrization Theorem of \cite{KAHL202247}—which states that every HDA is hhp-bisimilar to its symmetric version—into the presheaf framework, providing the formal basis for our order-free path semantics.
    \item \textbf{Forgetting Event Order in Path Semantics (Section \ref{sec: obs content}):} We introduce a labeling functor that assigns a canonical interval ipomset to each HDA path. This construction is intrinsic to the execution and eliminates representational artifacts without losing concurrent information.
    \item \textbf{Trace Correspondence (Section \ref{sec: relation}):} We establish a formal bridge between trace and ipomset-based observations, making ipomsets the faithful semantic counterpart to ST-traces.
    \item \textbf{Behavioral \& Logical Characterization (Section \ref{sec: Bis for HDA}):} We characterize ST- and hhp-bisimulation via ipomset isomorphism. Crucially, we show that our order-free semantics provides the necessary structure to apply the Open-Maps Framework canonically. This resolves previous ambiguities in the trace-based approach and ensures that the resulting modal logics correctly characterize hhp-bisimulation without representational artifacts.
\end{itemize}

\section{Base categories}\label{sec:base cat}
This section introduces the base categories underlying our variants of
higher-dimensional automata. They differ in how they represent
\emph{event ordering} and \emph{symmetry}.
The category $\square$ captures ordered concurrency through labeled event lists,
while $\Xi$ omits the event order, retaining only labels.
The generator based version $\Xi_g$ extends $\square$ by freely adding
permutation morphisms. We prove that $\Xi_g \cong \Xi$,
showing that adding event symmetries is equivalent to removing the event order.
After establishing this equivalence, we identify $\Xi_g$ with $\Xi$.
\subsection{The labeled precube category \texorpdfstring{$\square$}{[square]}}
Let $\Sigma$ be a given set of \emph{actions}.
\begin{definition}(\cite{struth2024presheaf})[Concurrency list]\label{def: conclist}
A \emph{concurrency list} (or \emph{conclist}) is a tuple $(U,\dashrightarrow,\lambda)$ where $U$ is a finite set equipped with a strict (i.e., irreflexive) total order $\dashrightarrow$, called the event order.
The function $\lambda:U\to\Sigma$ is a labeling map.

\end{definition}
The set $U$ models the concurrent local events active in a cell of an HDA, whereas the event order $\dashrightarrow$ can be seen as their index order.
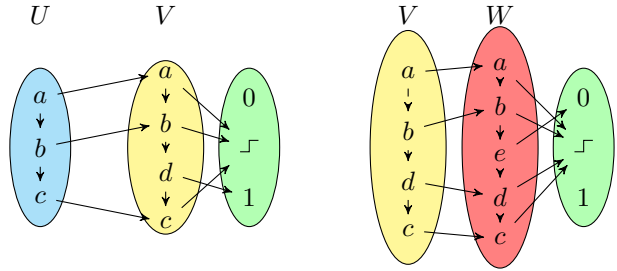
\begin{figure}[ht]
    \centering
    \begin{tikzpicture}[scale=1.1, every node/.style={font=\small}]
        % Draw set A
        \node[draw, ellipse, fill=cyan!30, minimum width=.8cm, minimum height=2.1cm, label=left:] (A) at (-1.5,0) {};
        \node at (-1.5, 0.6) (a) {$a$};
        \node at (-1.5, 0) (b) {$b$};
        \node at (-1.5, -0.6) (c) {$c$};
        \node at (-1.5, 1.6) (U) {$U$};
        \draw[->, dashed, bend left=0] (a) to (b);
        \draw[->, dashed, bend left=0] (b) to (c);

        % Draw set B
        \node[draw, ellipse, fill=yellow!50, minimum width=1cm, minimum height=2.3cm, label=right:] (B) at (0,0) {};
        \node at (0, 1.2-.6+.3) (1) {$a$};
        \node at (0, 0.4-.2+.1) (2) {$b$};
        \node at (0, -0.4+.1) (3) {$d$};
        \node at (0, -1.2+.6-.3) (4) {$c$};
        \node at (0, 1.6) (V) {$V$};
\draw[->, dashed, bend left=0] (1) to (2);
\draw[->, dashed, bend left=0] (2) to (3);
\draw[->, dashed, bend left=0] (3) to (4);
    % Draw set C
        \node[draw, ellipse, fill=green!30, minimum width=.8cm, minimum height=2.1cm, label=left:] (A) at (1,0) {};
        \node at (1, 0.6) (zero) {$0$};
        \node at (1, 0) (exec) {$\exec$};
        \node at (1, -0.6) (one) {$1$};

        %%%%%%%for the map g,zetaa
           % Draw set A
        \node[draw, ellipse, fill=yellow!50, minimum width=1cm, minimum height=3.1cm, label=right:] (Bp) at (1.5+1.4,0) {};
        %\node at (1.5+1, 0.6) (ap) {$a$};
        %\node at (1.5+1, 0) (bp) {$b$};
        %\node at (1.5+1, -0.6) (cp) {$c$};
        %\node at (1.5+1, 1.6) (Up) {$V$};
        \node at (1.5+1.4, 1.2-.3) (aa) {$a$};
        \node at (1.5+1.4, 0.4-.2) (bb) {$b$};
        \node at (1.5+1.4, -0.4+.2-.2) (cc) {$d$};
        \node at (1.5+1.4, -1.2+.2) (dd) {$c$};
        \node at (1.5+1.4, 1.6) (VV) {$V$};
\draw[->, dashed, bend left=0] (aa) to (bb);
\draw[->, dashed, bend left=0] (bb) to (cc);
\draw[->, dashed, bend left=0] (cc) to (dd);
        % Draw set B
        \node[draw, ellipse, fill=red!50, minimum width=1cm, minimum height=3.2cm, label=right:] (Bp) at (2+1+1,0) {};
        \node at (2+1+1, 1) (111) {$a$};
        \node at (2+1+1, 0.5) (222) {$b$};
        \node at (2+1+1, -0.4+.3) (eee) {$e$};
        \node at (2+1+1, -0.4-.2) (333) {$d$};
        \node at (2+1+1, -1.2+.1) (444) {$c$};
        \node at (2+1+1, 1.6) (VV) {$W$};

\draw[->, dashed, bend left=0] (111) to (222);
\draw[->, dashed, bend left=0] (222) to (eee);
\draw[->, dashed, bend left=0] (eee) to (333);
\draw[->, dashed, bend left=0] (333) to (444);
    % Draw set C
        \node[draw, ellipse, fill=green!30, minimum width=.8cm, minimum height=2.1cm, label=left:] (Ap) at (4+1,0) {};
        \node at (4+1, 0.6) (zerop) {$0$};
        \node at (4+1, 0) (execp) {$\exec$};
        \node at (4+1, -0.6) (onep) {$1$};

        % Draw arrows
        \draw[->] (a) -- (1);
        \draw[->] (b) -- (2);
       % \draw[->] (c) -- (3);
        \draw[->] (c) -- (4);
  \draw[->] (1) -- (exec);
   %\draw[->] (1) -- (2);
  \draw[->] (2) -- (exec);
  \draw[->] (4) -- (exec);
  \draw[->] (3) -- (one);
  %%%%%%%%%%%%%%%%%%%%
  %%%%%%%%%%%%%%%%%%%%arrows for g, zeta
        \draw[->] (aa) -- (111);
        \draw[->] (bb) -- (222);
       % \draw[->] (c) -- (3);
        \draw[->] (cc) -- (333);
        \draw[->] (dd) -- (444);
        \draw[->] (eee) -- (zerop);
  \draw[->] (111) -- (execp);
   %\draw[->] (1) -- (2);
  \draw[->] (222) -- (execp);
  \draw[->] (444) -- (execp);
  \draw[->] (333) -- (execp);
    \end{tikzpicture}
    \caption{Example of conclist maps $(f,\varepsilon):U \to V$ (on the left) and $(g,\zeta):V \to W$ (on the right).}
    \label{fig: exam of conclist maps}
\end{figure}
\begin{figure}[ht]
    \centering
    \begin{tikzpicture}[scale=1.1, every node/.style={font=\small}]

        % Draw set A
        \node[draw, ellipse, fill=cyan!30, minimum width=.8cm, minimum height=2.5cm, label=left:] (A) at (-1.5,0) {};
        \node at (-1.5, 0.6) (a) {$a$};
        \node at (-1.5, 0) (b) {$b$};
        \node at (-1.5, -0.6) (c) {$c$};
        \node at (-1.5, 1.6) (U) {$U$};
\draw[->, dashed, bend left=0] (a) to (b);
\draw[->, dashed, bend left=0] (b) to (c);
%\draw[->, dashed, bend left=0] (3) to (4);
        % Draw set B
        \node[draw, ellipse, fill=red!50, minimum width=1cm, minimum height=2.9cm, label=right:] (Bp) at (2-1+.5,0) {};
        \node at (2-1+.5, 1.2-.2) (111) {$a$};
        \node at (2-1+.5, 0.5) (222) {$b$};
        \node at (2-1+.5, -0.4+.4) (eee) {$e$};
        \node at (2-1+.5, -0.4-.1) (333) {$d$};
        \node at (2-1+.5, -1.2+.2) (444) {$c$};
        \node at (2-1+.5, 1.6) (VV) {$W$};
\draw[->, dashed, bend left=0] (111) to (222);
\draw[->, dashed, bend left=0] (222) to (eee);
\draw[->, dashed, bend left=0] (eee) to (333);
\draw[->, dashed, bend left=0] (333) to (444);
    % Draw set C
        \node[draw, ellipse, fill=green!30, minimum width=.8cm, minimum height=1.9cm, label=left:] (Ap) at (4+1,0) {};
        \node at (4+1, 0.6) (zerop) {$0$};
        \node at (4+1, 0) (execp) {$\exec$};
        \node at (4+1, -0.6) (onep) {$1$};

        % Draw arrows
        \draw[->] (a) -- (111);
        \draw[->] (b) -- (222);
  %%%%%%%%%%%%%%%%%%%%
  %%%%%%%%%%%%%%%%%%%%arrows for g, zeta
    %    \draw[->] (aa) -- (111);
      %  \draw[->] (bb) -- (222);
       % \draw[->] (c) -- (3);
        \draw[->] (c) -- (444);
      %  \draw[->] (dd) -- (444);
        \draw[->] (eee) -- (zerop);
  \draw[->] (111) -- (execp);
   %\draw[->] (1) -- (2);
  \draw[->] (222) -- (execp);
  \draw[->] (444) -- (execp);
  \draw[->] (333) -- (onep);
    \end{tikzpicture}
    \caption{Illustration of $(g,\zeta)\circ(f,\varepsilon):U \to W$, the composition of the conclist maps of Fig \ref{fig: exam of conclist maps}.}
        \label{fig: conclist composition}
\end{figure}
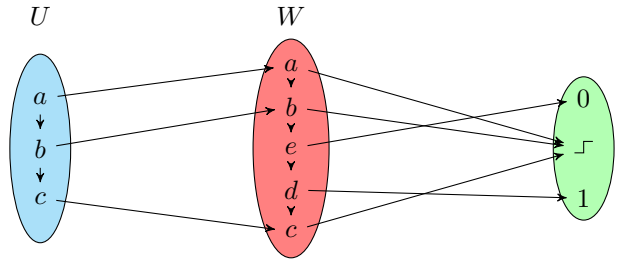

\begin{definition}(\cite{langthda})\label{def: conclist map} 
    A \emph{conclist map} from a conclist $U$ to $V$ is a pair $(f,\varepsilon)$
    such that:
    \begin{itemize}
     \item    $f:U\to V$ is a label and order-preserving function;
    \item 
        $\varepsilon: V\rightarrow \{ 0,\exec, 1 \}$ is a function such that $\varepsilon^{-1}(\exec)= f(U)$.
    \end{itemize}
    The composition of morphisms $(f,\varepsilon):S\to T$ and
    $(g,\zeta):T\to U$ is $(g,\zeta) \circ (f,\varepsilon)=(g\circ f,\eta)$, where
$\eta(u)=
\begin{cases}
\varepsilon(t) & \text{if } u=g(t)\text{ for some } t\in T,\\
\zeta(u) & \text{if } u\notin g(T).
\end{cases}$
  
    The \emph{full labeled precube category} $\squarefull$ has conclists as objects and conclist maps as morphisms.
\end{definition}
 The \emph{restricted} category $\square$ is the full subcategory on canonical conclists $(n,\lambda)$ with carrier $\{1\dashrightarrow\cdots\dashrightarrow n\}$ (and $(0)=\emptyset$). The inclusion $\square\hookrightarrow\squarefull$ is an equivalence; hence their presheaf categories are naturally equivalent \cite{LanguageofHDA}.
\paragraph{Intuition.}
Fix a conclist map $(f,\varepsilon):U\to V$ with $|U|=m$, $|V|=n$ ($m\le n$). Viewing $V$ as the list of $n$ potential concurrent events, $\varepsilon$ classifies each $v\in V$ as \emph{not yet started} ($0$), \emph{executing} ($\exec$), or \emph{terminated} ($1$). The map $f$ specifies how the $m$ active events of $U$ appear (by name and order) inside $V$; thus $f$ is determined by the inactive set $V\setminus\varepsilon^{-1}(\exec)$. See Fig \ref{fig: exam of conclist maps} and \ref{fig: conclist composition} for examples of conclist maps and their composition.
\paragraph{Isomorphisms.}
A conclist map $(f,\varepsilon):U\to V$ is an isomorphism iff $f$ is bijective; then necessarily $\varepsilon\equiv\exec$. We write $U\cong V$ for isomorphic conclists. If two conclists are isomorphic, then the isomorphism between them is unique \cite{KleeneTh}.
\paragraph{Insertion maps \texorpdfstring{$\iota_i$}{\iota\_i}.}
For canonical conclists $U=[n\!-\!1,\lambda]$ and $V=[n,\lambda']$, let $\iota_i:U\hookrightarrow V\; (1\le i\le n)$ be the unique order- and label-preserving injection whose image omits the $i$-th element of $V$. It has the following explicit expression
$$\iota_i(j)=
\begin{cases}
j & j<i,\\
j+1 & j\ge i.
\end{cases}$$
\paragraph{Canonical coface maps.} For canonical conclists $U=(n\!-\!1,\lambda)$ and $V=(n,\lambda')$ and $k\in\{0,1\}$, the coface map $d^k_{i,n}:U\to V$ $(1\le i\le n)$ is the unique conclist map $(\iota_i,\varepsilon)$ such that $\iota_i$ is the insertion map,
$\varepsilon(i)=k$, and $\varepsilon(j)=\exec\ (j\neq i)$. When $n$ is clear we write $d^k_i$.
\subsection{The symmetric labeled precube category \texorpdfstring{$\Xi$}{Xi}}
\begin{definition}[Concurrency set]\label{def: concset}
A \emph{concurrency set} (or \emph{concset}) is a pair $(U,\lambda)$ where $U$ is a finite set and $\lambda:U\to\Sigma$ is a labeling map.
That is, concsets are just conclists without an event order.
\end{definition}
A \emph{concset map} from $U$ to $V$ is a pair $(f,\varepsilon)$ such that
\begin{itemize}
  \item $f:U\to V$ is an injective, label-preserving function;
  \item $\varepsilon:V\to\{0,\exec,1\}$ satisfies $\varepsilon^{-1}(\exec)=f(U)$.

  Composition of concset maps is defined as in Definition~\ref{def: conclist map}.
\end{itemize}
Write $\Xifull$ for the full symmetric labeled precube category on all concsets and concset maps. The \emph{restricted} subcategory $\Xi$ has as objects the canonical concsets $(n,\lambda)$ with carrier $\{1,\dots,n\}$ (and $(0)=\emptyset$). As previously, the inclusion $\Xi\hookrightarrow\Xifull$ is an equivalence. 
\paragraph{Interpretation.}
A similar interpretation applies here; given a concset map $(f,\varepsilon):U\to V$ with $|U|=m$ and $|V|=n$ $(m\le n)$, the function $\varepsilon:V\to\{0,\exec,1\}$ classifies each event of $V$ as \emph{not yet started} $(0)$, \emph{executing} $(\exec)$, or \emph{terminated} $(1)$, and $f:U\hookrightarrow V$ embeds the $m$ active events by ensuring $\varepsilon^{-1}(\exec)=f(U)$. However, unlike the ordered case $\square$, the absence of an event order in $\Xi$ means that $\varepsilon$ does not need to determine $f$.

\emph{Example.} From $(2,aa)$ to $(3,aab)$, take $\varepsilon=(\exec,\exec,0)$ on the target. There are two distinct label-preserving injections $f_1,f_2$ that map the two $a$'s of the source to the two $a$'s of the target (identity vs.\ swap). Thus the same $\varepsilon$ admits multiple $f$'s. Consequently, concset isomorphisms are not unique; in fact, any permutation of equally labeled events yields another isomorphism.
\paragraph{Forgetting order.}
There is an obvious forgetful functor $s:\square \to \Xi$ that erases the event order on objects and sends a conclist map $(f,\varepsilon)$ to the same pair viewed as a concset map. Thus $\square$ is a wide subcategory of $\Xi$. The inclusion is not full (e.g.\ nontrivial permutations exist in $\Xi$ but not in $\square$). We henceforth identify $\square$ with its image under $s$ and write $\square\subseteq \Xi$.
\paragraph{Canonical coface maps.}
For canonical concsets $U=(n\!-\!1,\gamma)$ and $V=(n,\gamma')$ we define analogously $d^k_{i,n}:U\to V$ with the same $\varepsilon$ and insertion map $\iota_i$ (now without the event order).

\subsection{The generator based categories \texorpdfstring{$\Xi_g$}{Xi	extsubscript{g}} and \texorpdfstring{$\square_g$}{square	extsubscript{g}}}\label{sub: xi_g} 

This section introduces the base categories for our HDA variants, which differ in their treatment of \emph{event ordering} and \emph{symmetry}. While $\square$ represents ordered concurrency, $\Xi$ omits event order to provide a coordinate-free foundation. We introduce $\Xi_g$ by freely adding permutation morphisms to $\square$ and prove $\Xi_g \cong \Xi$. This isomorphism demonstrates that \textbf{adding symmetries is equivalent to removing event order}, allowing us to identify $\Xi_g$ with $\Xi$ henceforth.

\paragraph{Symmetric group and induced permutations on faces.}
For $n\ge 0$ let $\mathfrak{S}_n$ be the symmetric group on $\{1,\dots,n\}$; write
$\mathfrak{S}=\bigsqcup_{n\ge 0}\mathfrak{S}_n$. For $n\ge 1$, $\theta\in\mathfrak{S}_n$, and $i\in\{1,\dots,n\}$, define the
\emph{induced permutation on faces} $d_i\theta\ \in\ \mathfrak{S}_{n-1}$ as the unique permutation making the square commute:
\begin{equation}
\label{eq: Perm}
\hspace*{-\mathindent}\hfill
\begin{tikzcd}[column sep=large,row sep=large]
{[n-1]} \arrow[r, "d_i\theta"] \arrow[d, hook, "\iota_{\theta^{-1}(i)}"'] 
  & {[n-1]} \arrow[d, hook, "\iota_{i}"] \\
{[n]}   \arrow[r, "\theta"'] 
  & {[n]}
\end{tikzcd}
\quad\Longleftrightarrow\quad
\iota_i\circ d_i\theta \;=\; \theta\circ \iota_{\theta^{-1}(i)}.
\hfill
\end{equation}
Intuitively, $d_i\theta$ is obtained from $\theta$ by deleting position
$\theta^{-1}(i)$ in the domain and $i$ in the codomain and renumbering.
It can be expressed componentwise as $$d_i\theta(j)= \begin{cases} \theta(j), & j<\theta^{-1}(i),~\theta(j)< i,\\[2pt] \theta(j)-1, & j<\theta^{-1}(i),~\theta(j)> i,\\[2pt] \theta(j+1), & j\ge \theta^{-1}(i),~\theta(j+1)< i,\\[2pt] \theta(j+1)-1, & j\ge \theta^{-1}(i),~\theta(j+1)> i. \end{cases} $$

\begin{definition}\label{def: Xi_g}
The category $\Xi_g$ is presented as follows.
\begin{enumerate}
\item \textbf{Objects.}
  Canonical conclists $(n,\lambda)$ with $n\ge 0$ and a labeling map $\lambda:(n)\to\Sigma$.
  % (surjectivity is \emph{not} required)

\item \textbf{Generating morphisms.}
  \begin{enumerate}
  \item \emph{Coface maps.}
    For each $n\ge 1$, $1\le i\le n$, $k\in\{0,1\}$ and label
    $\lambda:(n)\to\Sigma$, let $d^{k}_{i}\;:\;(n-1,\,\lambda\circ\iota_{i})
      \;\to\;(n,\,\lambda),$
  \item \emph{Symmetry maps.}
    For each $n\ge 1$ and $\tau\in\SG_{n}$, let $ \tau:(n,\,\lambda)\to (n,\,\lambda\circ\tau^{-1}).$
  \end{enumerate}

\item \textbf{Relations.}
  \begin{enumerate}\setlength\itemsep{4pt}
  \item \emph{Cubical (face) identities.}
    For $i<j$, $d^{\ell}_{j}\,d^{k}_{i}\;=\;d^{k}_{i}\,d^{\ell}_{j-1}
      \;:\;(n-2,\,\lambda\circ\iota_{i}\circ\iota_{j-1})
      \to(n,\,\lambda).$ 
  \item \emph{Group identities.}\footnote{For brevity we often write simply $\id$ instead of
$\id_{(n,\lambda)}$; the intended domain and codomain are always clear
from the context.}
    $\mathrm{\id}_{(n,\lambda)}$ is the identity and
    $\tau\circ\sigma=(\tau\sigma)$. %in $\mathfrak S_{n}$.
  \item \emph{Permutation–face interchange.}
    For each $\tau\in\mathfrak S_{n}$ and $1\le i\le n$, $d^k_i\circ d_i\theta = \theta \circ d^k_{\theta^{-1}(i)}$
%    =============== 
  \end{enumerate}
\end{enumerate}
The category $\Xi_g$ is obtained by freely adding identities and
composites to these generators and quotienting by the stated relations. 
\end{definition}

We define the category $\square_g$ as the subcategory of $\Xi_g$ obtained by
forgetting about permutations, that is, by removing the symmetry generators
and the corresponding relations. Concretely, $\square_g$ is defined by
omitting Items~2(b), 3(b), and~3(c) in Definition~\ref{def: Xi_g}.
\subsection{Relation between the base categories}
The goal of this paragraph is to establish relations between precubical categories we introduced so far and that we use in later sections. The goal is to show that these categories fit into the diagram 
$$\begin{tikzcd}
    \square_{full} \ar[d]
    &
    \square \ar[l,"\supseteq" swap]\ar[r,"\cong"]\ar[d,"\subseteq"]
    &
    \square_g \ar[d]
\\
    \Xi_{full} 
    &
    \Xi \ar[l,"\supseteq" swap]\ar[r,"\cong"]
    &
    \Xi_g
\end{tikzcd}$$
We want to prove that the left horizontal functors are natural equivalences (this part is fine) and that the right horizontal functors are isomorphisms. The latter allows us to identify $\Xi$ and $\Xi_g$ and use generators-and-relations and concset maps interchangeably.

Below is the construction of the functor $F:\Xi_g\to\Xi$ and the proof that it is an isomorphism.

\begin{lemma}
    There exists a unique functor $F:\Xi_g\to\Xi$ such that
    \begin{itemize}
        \item $F((n,\lambda))=(n,\lambda)$ (the identity on objects).
        \item $F(d^k_i:(n,\lambda)\to(n-1,\lambda \circ \iota_i))=(\iota_i,\varepsilon_i^k)$, where $\varepsilon^k_i(j)=k$ if $j=i$ and $\exec$ otherwise.
        \item $F(\tau)=(\tau,const_{\exec})$.
    \end{itemize}
\end{lemma}
\begin{proof}
Since $\Xi_g$ is presented by generators (face maps and permutations) and
relations, it suffices to check that the assignment $F$ preserves the
defining relations.

\emph{Cubical relations.}
For $i<j$, the relation
$d_j^{\ell}d_i^{k}=d_i^{k}d_{j-1}^{\ell}$ holds in $\Xi_g$.
Applying $F$ yields $(\iota_j,\varepsilon_j^{\ell})\circ(\iota_i,\varepsilon_i^{k})
\text{ and }
(\iota_i,\varepsilon_i^{k})\circ(\iota_{j-1},\varepsilon_{j-1}^{\ell}),$ which coincide since $\iota_j\circ\iota_i=\iota_i\circ\iota_{j-1}$ and the
resulting $\varepsilon$-components agree pointwise by inspection.

\emph{Group relations,}
for $\sigma,\tau\in S_n$, functoriality follows immediately from $(\sigma,\mathrm{const}_{\exec})\circ(\tau,\mathrm{const}_{\exec})
=(\sigma\circ\tau,\mathrm{const}_{\exec}),$ and identities are preserved.

\emph{Permutation–face interchange.}
For $\theta\in S_n$ and $i$, both composites
$F(\theta)\circ F(d_{\theta^{-1}(i)}^{\,k})$ and
$F(d_i^{\,k})\circ F(d_i\theta)$ have first component
$\theta\circ\iota_{\theta^{-1}(i)}=\iota_i\circ d_i\theta$.
A direct application of the composition rule in $\Xi$ shows that their
$\varepsilon$-components coincide, taking value $k$ at $i$ and $\exec$
elsewhere. Thus $F$ preserves all defining relations and induces a well-defined functor
$F:\Xi_g\to\Xi$.

\emph{Uniqueness.}
Since $\Xi_g$ is presented by generators and relations, $F$ is uniquely
determined by its action on generators.
\end{proof}

\begin{definition}
    The \emph{canonical presentation} of a morphism $v$ of $\Xi_g$ is an equation $v=d_{i_r}^{k_r}\circ \dotsm\circ d_{i_2}^{k_2}\circ d_{i_1}^{k_1}\circ \tau$ such that $i_1<\dotsm<i_r$.
\end{definition}

\begin{lemma}
    Every morphism $v$ of $\Xi_g$ has a canonical presentation.
\end{lemma}
\begin{proof}
%    (sketch) Let $v=u_1\circ\dotsm\circ u_r$ be any presentation with generators $u_j$. We make it "more canonical" using the following operations:
 %   \begin{itemize}
  %      \item If there are two consecutive permutations, we merge them using \ref{def: Xi_g}.3.(b).
   %     \item If a permutation precedes a coface map, we swap them using \ref{def: Xi_g}.3.(c). %(which is wrong, should be $d^k_i\circ d_i\theta = \theta \circ d^k_{\theta^{-1}(i)}$ as in \eqref{eq: Perm}).
    %    \item If we have two consecutive coface maps $d^k_i\circ d^\ell_j$ and $j\geq i$, then we switch them using \ref{def: Xi_g}.(a).
   % \end{itemize}
   % ==============================
Fix any presentation $v=u_1\circ\cdots\circ u_m$ over the generators
$\{d_i^{\,k}\}\cup\{\theta\in\mathfrak S_*\}$. We construct, by induction on $m$, a factorisation $\mathrm{NF}(v)\;=\;D\circ\tau
$ with $
D=d_{i_r}^{k_r}\circ\cdots\circ d_{i_1}^{k_1},$ $ i_1<\cdots<i_r,\ \ \tau\in \mathfrak S_*,$ and show that $\mathrm{NF}(v)$ equals $v$ in $\Xi_g$ (i.e.\ is obtained by applying only the defining relations).

\smallskip\noindent\emph{Base case $m=0$.}
$\mathrm{NF}(\mathrm{id})=\mathrm{id}\circ \mathrm{id}$ is canonical.

\smallskip\noindent\emph{Inductive step.}
Write $v=w\circ g$ with $g$ a generator and suppose
$\mathrm{NF}(w)=D\circ\sigma$ is canonical, where $D=d_{i_s}^{\ell_s}\circ\cdots\circ d_{i_1}^{\ell_1}$ with $ \sigma\in\mathfrak S_*.$ We define $\mathrm{NF}(v)$ by cases on $g$.

\medskip\noindent\emph{Case 1: $g=\theta\in\mathfrak S_*$ (a permutation).}
Put $\mathrm{NF}(w\circ \theta)\;:=\;D\circ(\theta\,\sigma).$ This uses only the group law (Def.\ \ref{def: Xi_g}.3(b)) and clearly preserves canonicity of the coface block.

\emph{Case 2: $g=d_i^{\,k}$ (a coface).}
First \emph{transport $d_i^{\,k}$ through the current permutation $\sigma$ to the left} using the permutation–face interchange in the form $\sigma\circ d_i^{\,k} \;=\; d_{\sigma(i)}^{\,k}\circ d_{\sigma(i)}\sigma,$ which follows from
$d^k_j\circ d_j\theta = \theta \circ d^k_{\theta^{-1}(j)} $
by setting $\theta=\sigma$, $j=\sigma(i)$.
Hence $D\circ \sigma\circ d_i^{\,k}
\;=\;
D\circ d_{\sigma(i)}^{\,k}\circ (d_{\sigma(i)}\sigma).$ Set $\sigma':=d_{\sigma(i)}\sigma\in\mathfrak S_{*-1}$ (the induced permutation on one fewer coordinate). It remains to \emph{insert} the single coface $d_{\sigma(i)}^{\,k}$ into the sorted block $D$.

\smallskip\noindent\underline{Insertion lemma.}
Given a sorted block $D=d_{i_s}^{\ell_s}\circ\cdots\circ d_{i_1}^{\ell_1}$ with $i_1<\cdots<i_s$ and a coface $d_j^{\,k}$ on its right, there exists a finite sequence of applications of the cubical identity
\[ \qquad \qquad \qquad \qquad \qquad \qquad
d_{q}^{\lambda}\circ d_{p}^{\kappa} \;=\; d_{p}^{\kappa}\circ d_{q-1}^{\lambda}\qquad(p<q)
\tag{$\flat$}
\]
that rewrites $D\circ d_j^{\,k}$ into a sorted block $D'$ $D\circ d_j^{\,k}\;\equiv\; d_{i'_{s'}}^{\ell'_{s'}}\circ\cdots\circ d_{i'_1}^{\ell'_1}
\quad$ $i'_1<\cdots<i'_{s'},$ where each swap with a left neighbour $d_p^{\kappa}$ (with $p<j$) replaces the pair by $d_p^{\kappa}\circ d_{j-1}^{\,k}$, i.e.\ decrements the index of the moving coface by $1$. Concretely: starting from the right, compare $j$ with $i_s$; while $j\le i_t$, apply $(\flat)$ to swap with $d_{i_t}^{\ell_t}$ and update $j\gets j-1$; when $j>i_t$, stop and place $d_j^{\,k}$ immediately to the left of $d_{i_t}^{\ell_t}$. This terminates since $j$ strictly decreases on each swap.

Applying the insertion lemma with $j=\sigma(i)$ yields a sorted block $D'$ such that $D\circ d_{\sigma(i)}^{\,k}\;\equiv\;D'$ with indices strictly increasing. Define $\mathrm{NF}(w\circ d_i^{\,k})\;:=\;D'\circ\sigma'.$ By construction we used only the interchange law and the cubical identities, hence $\mathrm{NF}(w\circ d_i^{\,k})$ is equal to $w\circ d_i^{\,k}$ in $\Xi_g$, and its coface block is sorted. In both cases we obtain a canonical form for $v$, completing the induction. Therefore every morphism $v$ of $\Xi_g$ admits a presentation of the form
$d_{i_r}^{k_r}\circ\cdots\circ d_{i_1}^{k_1}\circ\tau$ with $i_1<\cdots<i_r$.
\end{proof}
\begin{lemma}
    $F$ is surjective.
\end{lemma}
\begin{proof}
 Let $(f,\varepsilon):(m,\lambda)\to(n,\mu)$ be a morphism in $\Xi$.
Set $J:=f([m])\subseteq(n)$ and $I:=(n)\setminus J=\{i_1<\cdots<i_r\}$ ($r=n-m$).
Let $\iota_J:J\hookrightarrow(n)$ be the inclusion and let
$\rho:(m)\xrightarrow{\cong}J$ be the unique increasing bijection
($J=\{j_1<\cdots<j_m\}$, so $\rho(t)=j_t$).
Define the permutation $\tau := \rho^{-1}\circ f \in \mathfrak{S}_m,
\qquad\text{so that } f = \iota_J \circ \rho \circ \tau.$ Then $\tau$ is the unique permutation satisfying
$\tau(i)<\tau(i')$ iff $f(i)<f(i')$;
equivalently, $f\circ\tau^{-1}$ is order-preserving. Define in $\Xi_g$ $\alpha\ :=\ d_{i_r}^{k_r}\circ\cdots\circ d_{i_1}^{k_1}\circ\tau
\text{ with }k_s:=\varepsilon(i_s)\in\{0,1\}.$

\emph{First component.}
Write $\iota_I:=\iota_{i_r}\circ\cdots\circ\iota_{i_1}:(m)\hookrightarrow(n)$.
\emph{Claim.} $\iota_I=\iota_J\circ\rho$.
(Induction on $r$: composing the standard injections that skip $I$ yields the
order-preserving bijection $(m)\to J$ followed by inclusion.) This is because for each $t\in(m)$, by definition $\iota_I(t)$ is the $t$-th smallest
element of $J$, i.e. $\iota_I(t)=j_t$. Since $\iota_J$ is the inclusion of $J$,
$(\iota_J\circ\rho)(t)=\iota_J(j_t)=j_t=\iota_I(t)$, for all $t$. Hence
$\iota_I=\iota_J\circ\rho$. Hence the first component of $F(\alpha)$ equals
$$\iota_{i_r}\circ\cdots\circ\iota_{i_1}\circ\tau
=\iota_I\circ\tau
=(\iota_J\circ\rho)\circ\tau
=f \quad\text{by }(*).$$

\emph{Second component.}
Using $F(d_{i_s}^{k_s})=(\iota_{i_s},\varepsilon_{i_s}^{k_s})$ and the
composition rule in $\Xi$, a trivial induction on $r$ gives after composing
the $r$ cofaces
\[
\tilde\varepsilon(v)=
\begin{cases}
k_s & \text{if } v=i_s\in I,\\
\exec & \text{if } v\in J.
\end{cases}
\]
Composing with $F(\tau)=(\tau,\mathrm{const}_{\exec})$ does not change
$\tilde\varepsilon$ (the status $\exec$ is constant and the injection image is $J$).
Thus the second component of $F(\alpha)$ is
$v\longmapsto
\begin{cases}
\varepsilon(i_s)=k_s & \text{if } v=i_s\in I,\\
\exec & \text{if } v\in J=f([m]),
\end{cases}$ which is exactly $\varepsilon$. Therefore $F(\alpha)=(f,\varepsilon)$, so $F$ is surjective.
\end{proof}
\begin{lemma}
    $F$ is injective. 
\end{lemma}
\begin{proof}
 %   (sketch) Essentially we need to prove that if $\alpha$ and $\alpha'$ have different canonical presentations, then $F(\alpha)\neq F(\alpha')$. First we show that $\tau$'s must be the same, then $i_s$'s and $k_s$'s. 
    Let $\alpha,\alpha'$ be morphisms with canonical presentations $\alpha \;=\; d_{i_r}^{k_r}\circ\cdots\circ d_{i_1}^{k_1}\circ\tau,$ and $\alpha' \;=\; d_{i'_{r'}}^{k'_{r'}}\circ\cdots\circ d_{i'_1}^{k'_1}\circ\tau',$ where $i_1<\cdots<i_r$ and $i'_1<\cdots<i'_{r'}$.
Assume $F(\alpha)=F(\alpha')=(f,\varepsilon)$; we show the canonical data coincide. \emph{Step 1 (recover $I$ and $k$’s from $\varepsilon$).}
In $\Xi$ we have $\varepsilon^{-1}(\exec)=f([m])$. Hence $I \;:=\; (n)\setminus f([m]) \;=\; \{v\in(n)\mid \varepsilon(v)\neq \exec\}.$ But for a canonical word, $I=\{i_1<\cdots<i_r\}$ and $\varepsilon(i_s)=k_s$.
Thus $I$ and the tuple $(k_s)_{s=1}^r$ are uniquely determined by $(f,\varepsilon)$.
Applying the same to $\alpha'$ yields $I'=\{i'_1<\cdots<i'_{r'}\}=I$ and $k'_s=k_s$,
so $r'=r$, $i'_s=i_s$ and $k'_s=k_s$ for all $s$.

\emph{Step 2 (recover $\tau$ from $f$).}
Let $J:=f([m])$ and $\rho:(m)\xrightarrow{\cong} J$ be the increasing bijection.
By the insertion-of-holes identity $\iota_{i_r}\circ\cdots\circ\iota_{i_1}
=\iota_J\circ\rho$, the first component of $F(\alpha)$ is $f \;=\; (\iota_J\circ\rho)\circ\tau.$ Hence $\tau=\rho^{-1}\circ f$, which is uniquely determined by $f$. The same computation for $\alpha'$ gives $\tau'=\rho^{-1}\circ f=\tau$. We have shown that different canonical data cannot map to the same $(f,\varepsilon)$:
from $F(\alpha)=F(\alpha')$ we obtained $r=r'$, $i_s=i'_s$, $k_s=k'_s$, and $\tau=\tau'$.
Thus $\alpha$ and $\alpha'$ have the same canonical presentation, hence are equal in $\Xi_g$.
Therefore $F$ is injective.
\end{proof}
Once surjectivity and injectivity have been established, we obtain the following result.
\begin{proposition}\label{prop:base-iso}
The functor $F:\Xi_g\to\Xi$ is an isomorphism of categories.
\end{proposition} Consequently, the subcategories obtained by forgetting permutations are also
isomorphic, that is, $\square_g \;\cong\; \square.$ \textbf{From now on, we simply write $\Xi=\Xi_g$ and $\square=\square_g$}, and freely use
their descriptions by generators and relations.
The base categories that will be used throughout the remainder of the paper are
summarized in Table~\ref{tab:base-categories-lean}.

\begin{table}[t]
  \centering
  \small
  \renewcommand{\arraystretch}{1.2}
  \begin{tabular}{@{}l p{0.28\linewidth} p{0.55\linewidth}@{}}
    \toprule
    \textbf{Category} & \textbf{Objects} &
      \textbf{Morphisms \& generators} \\
    \midrule
    $\square$ &
      canonical conclists &
      conclist maps; generated by coface maps\\
    $\Xi$ &
      canonical concsets &
      concset maps; generated by coface maps and permutations\\
    \bottomrule
  \end{tabular}
  \caption{Base categories at a glance.}
  \label{tab:base-categories-lean}
\end{table}
\section{Precubical sets and HDA}\label{sec:HDA}
We now turn to the presheaf formulation of HDAs over these bases. We establish that the symmetric HDAs of Kahl~\cite{KAHL202247} (based on $\Xi_g$) and those of Struth and Ziemiański~\cite{struth2024presheaf} (based on $\Xi$) are categorically isomorphic. \textbf{Crucially, this equivalence allows us to leverage Kahl’s Symmetrization Theorem}---the result that every HDA is hhp-bisimilar to its symmetric expansion. This justifies our order-free semantics as a behavioral model for the \textbf{entire class of HDAs, not only the symmetric ones}.
\subsection{Precubical sets over a base}\label{subsec:precubical-over-base}

\begin{definition}[Precubical set over a base]\label{def:precubical-sets-general}
Let $\mathbb{C}\in\{\square,\;\Xi\}$ be one of the base categories introduced above.
A \emph{precubical set over $\mathbb{C}$} is a presheaf $X:\mathbb{C}^{\mathrm{op}}\to\mathbf{Set}.$ A \emph{map of precubical sets} $g:X\Rightarrow Y$ over $\mathbb C$
is a natural transformation.
We write $\widehat{\mathbb{C}}$ for the corresponding presheaf category.
\end{definition}
One can equally define presheaves over the symmetric base $\Xi_g$.
However, since we have shown that $\Xi_g$ and $\Xi$ are isomorphic,
these bases give rise to equivalent presheaf categories, as stated below.
\begin{proposition}[Equivalence of presheaf categories]\label{prop:cat-isomorphism}
There is an isomorphism of presheaf categories $\widehat{\Xi} \;\cong\; \widehat{\Xi}_g.$
\end{proposition}
\begin{proof}
Since by Prop~\ref{prop:base-iso} $F : \Xi_g \xrightarrow{\;\cong\;} \Xi$ is an isomorphism of categories, the induced pullback and pushforward along~$F$ yield an isomorphism of presheaf categories
$\widehat{\Xi} \cong \widehat{\Xi_g}$.
\end{proof}
\paragraph{Terminology and correspondence.}
\begin{itemize}
  \item \textbf{Case $\mathbb C=\square$.}
  This yields the classical notion of a \emph{precubical set}
  in the sense of van Glabbeek~\cite{VANGLABBEEK2006265} and
  Fahrenberg et al.~\cite{LanguageofHDA,KleeneTh},
  based on ordered concurrency lists.

  \item \textbf{Case $\mathbb C=\Xi$.}
  This gives rise to the \emph{symmetric precubical sets}
  of Kahl~\cite{KAHL202247} and of Struth and Ziemiański~\cite{struth2024presheaf},
  where concurrent events are represented without an explicit order.
  Since $\Xi\simeq\Xi_g$, both constructions define the same class
  of \emph{symmetric precubical sets}.
\end{itemize}

\paragraph{Notation.}
\begin{itemize}
    \item For any coface morphism $d^k_{i,n}:U\to V$, with $|U| = n$, in the chosen base $\mathbb C\in\{\square,\Xi\}$ we set $\delta^k_{i,n}\;:=\;X[d^k_{i,n}],$ and this notation is used uniformly for all bases.
    \item Given an $n$-cell $x$ and a permutation $\tau\in\mathfrak S_n$, we write $\tau\cdot x$ for $X[\tau](x) \in X(n,\lambda\circ\tau^{-1})$ 
\end{itemize}
If \( X \) is a presheaf over a category \(\mathbb{C}\), where \(\mathbb{C}\) is either \(\square\) or \(\Xi\), and \( U \) is an object $X[U]$ 
represents the set of those $n$-cells in $X$ in which exactly the events in $U$ are active. 
Elements of \( X[U] \), form the set of cells \(\cell_X\) of \( X \). Specifically, we have:
$$\cell_X = \bigsqcup_{U \in \text{obj}(\mathbb{C})} X[U].$$
For any $x \in X[U],$ elements of $U$ are called events of $x$. We write $\ev(x)=U$.

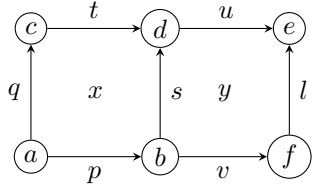
\begin{figure}[ht]
    \centering
     \setlength{\tabcolsep}{.2cm}
  \begin{tabular}{ccc}
 \begin{tikzpicture}[x=1.7cm, y=1.7cm, >=stealth]
  %--- node (state) style -------------------------------------------------------
  \tikzset{state/.style={circle, draw, fill=white, inner sep=2pt, font=\small}}

  %--- vertices -----------------------------------------------------------------
  \node[state] (a) at (0,0) {$a$};
  \node[state]                         (b) at (1,0) {$b$};
  \node[state, label=above:] (c) at (0,1) {$c$};
  \node[state, label=above:] (d) at (1,1) {$d$};
  \node[state, label=above:] (e) at (2,1) {$e$};
   \node[state, label=above:] (f) at (2,0) {$f$};

  %--- shaded square (concurrent region) ---------------------------------------
  \fill[black!15] (a) -- (b) -- (d) -- (c) -- cycle;

  %--- arrows with labels -------------------------------------------------------
  \draw[->] (a) -- node[below] {\(p\)} (b);
  \draw[->] (b) -- node[right] {\(s\)} (d);
  \draw[->] (a) -- node[left ] {\(q\)} (c);
   \draw[->] (f) -- node[right] {\(l\)} (e);
  \draw[->] (c) -- node[above] {\(t\)} (d);
  \draw[->] (d) -- node[above] {\(u\)} (e);
  \draw[->] (b) -- node[below] {\(v\)} (f);

  %--- 2-cell label -------------------------------------------------------------
  \node at ($(a)!0.5!(d)$) {\(x\)};
    \node at ($(b)!0.5!(e)$) {\(y\)};
\end{tikzpicture}

&
 \end{tabular}
  \bigskip
    \caption{Precubical set $X$ with two 2-dimensional cells}
    \label{fi: face maps}
\end{figure}
\begin{example}[Example of a precubical set] Figure \ref{fi: face maps} shows a visualization of a precubical set $X$, in which all one dimensional cells have the same label. So we ignore $\lambda$ in the conclist notation and write $(n)$. The precubical set $X$ is given by the following:
\[
X[0] = \{a,b,c,d,e,f\},\quad X[1] = \{p,q,s,t,u,v,l\},\quad X[2] = \{x,y\},
\]
and face maps:
\[
\begin{aligned}
&\delta^0_1(p) = a,\ \delta^1_1(p) = b, \quad
\delta^0_1(q) = a,\ \delta^1_1(q) = c, \quad
\delta^0_1(u) = d,\ \delta^1_1(u) = e, \\
&\delta^0_1(s) = b,\ \delta^1_1(s) = d, \quad
\delta^0_1(t) = c,\ \delta^1_1(t) = d, 
\quad ~~~
\delta^0_1(v) = b,\ \delta^1_1(v) = f,\\
&\delta^0_2(x) = q,\ \delta^1_2(x) = s, \quad
\delta^0_1(x) = p,\ \delta^1_1(x) = t, 
\\
&\delta^0_2(y) = s,\ \delta^1_2(y) = l, \quad
\delta^0_1(y) = v,\ \delta^1_1(y) = u.
\end{aligned}
\]
In the figure, $\Sigma$ has one element and the indices refer to the names of the cells.
\end{example}
% Note that a precubical set may consist of a single 2–cell together with its faces and vertices.
% This is no longer possible in the symmetric setting.
% {\color{brown} Unless the labels are the same.}

\begin{example}\label{ex:single-2d-cell}
Let \(Y\in\widehat{\square}\) be the precubical set with a single 2–cell
$x\in Y[(2,ab)]$, and no cells in higher dimensions.
Such a configuration cannot occur in a symmetric precubical set.
Indeed, the transposition
\(\tau : (2,ab)\xrightarrow{\;\cong\;}(2,ba)\)
is a morphism in~\(\Xi\), so for any \(Z\in\widehat{\Xi}\),
the induced map
\(Z[\tau]:Z(2,ab)\to Z(2,ba)\)
is a bijection.
Hence, the existence of a cell in \(Z(2,ab)\) forces the existence of a corresponding cell in \(Z(2,ba)\). This obstruction persists in the autoconcurrent case \(a=b\).
If \(x\) were the only 2–cell, symmetry would imply \(\tau\cdot x=x\) for the transposition \(\tau=(2\;1)\),
yielding $d^0_1 x = d^0_1(\tau\cdot x) = d^0_{\tau^{-1}(1)}x = d^0_2 x,$ a contradiction.
\end{example}
\begin{definition}[Higher-dimensional automaton over a base]\label{def:HDA-over-base}
Let $\mathbb C\in\{\square,\Xi\}$.
An \emph{HDA over $\mathbb C$} is a pair $\mathcal X=(X,i_X)$ where:
\begin{itemize}
  \item $X\in\widehat{\mathbb C}$ is a precubical set over~$\mathbb C$;
  \item $i_X\in X[0]$ is the initial $0$–cell.
\end{itemize}
For a natural transformation $f:X\Rightarrow Y$, we write
$f_U : X[U] \to Y[U]$ for its component at $U\in\mathbb C$.
Given HDAs $\mathcal X=(X,i_X)$ and $\mathcal Y=(Y,i_Y)$ over $\mathbb C$, a \emph{morphism}
$\mathcal X\to\mathcal Y$ is a natural transformation $f:X\Rightarrow Y$ such that
its $0$-component preserves the chosen initial cell, i.e.\ $f_0(i_X)=i_Y$.

HDAs over $\square$ coincide with the classical (ordered) HDAs,
while HDAs over $\Xi$ are the symmetric HDAs.
\end{definition}
Final cells are omitted here without loss of generality:
HDAs need not have designated final states of these can be added as an external structure when modeling successful termination, but are irrelevant for the properties studied here.
\begin{definition}[Symmetriser]\label{def:symmetriser}(\cite{KAHL202247})
Let $X\in\widehat{\square}$ be a precubical set.
The \emph{free symmetric precubical set} generated by~$X$
is the presheaf $SX\in\widehat{\Xi}$ defined as follows:
\begin{itemize}
  \item \emph{On objects.} For each canonical conclist $(n,\lambda)$,
  $$  (SX)(n,\lambda)
    \;=\;
    \{\,(\theta,x)\mid \theta\in\mathfrak S_n,\;
       x\in X(n,\lambda\circ\theta)\,\}.$$
  \item \emph{On cofaces.} For
  $d_i^{\,k}:(n\!-\!1,\lambda\circ\iota_i)\to(n,\lambda)$, $$ (SX)[d_i^{\,k}](\theta,x)\;=\;\bigl(d_i\theta,\;X[d^{\,k}_{\theta^{-1}(i)}](x)\bigr),$$
  where $d_i\theta\in\mathfrak S_{n-1}$ is the induced permutation on faces defined in~(\eqref{eq: Perm}).
  \item \emph{On permutations.} For $\tau\in\mathfrak S_n$,
     $$(SX)[\tau](\theta,x)
     \;=\;
     (\tau\theta,\;x)
     \;:\;
     (SX)(n,\lambda)\longrightarrow (SX)(n,\lambda\circ\tau^{-1}).$$
\end{itemize}
The above assignments respect the cubical identities and the
permutation–face interchange law of~$\Xi$,
so they extend uniquely to a functor $SX:\Xi^{\mathrm{op}}\to\mathbf{Set}$. We call $SX$ the \emph{symmetrisation} of~$X$.
For an HDA $\mathcal X=(X,i_X)$ over $\square$,
the \emph{free symmetric HDA} is $$ S\mathcal X := (SX,\,(\id.i_X))$$

\end{definition}
For each $n$–cell $x\in X[n,\lambda]$ and $\tau\in\mathfrak S_n$,
the element $(\theta,x)\in (SX)(n,\lambda\circ\theta^{-1})$
represents the same $n$ concurrent events as $x$, but with their linear order
permuted by~$\theta$.
In the symmetric interpretation, this element corresponds to the cell $\tau.x$.
Hence, the symmetrisation $SX$ contains all cells of $X$ together with
their symmetric variants—one for each permutation of their $n$ events.
Consequently, every $n$–cell of $X$ gives rise to $n!$ distinct copies in $SX$. %See Figure \ref{fig:hda and its symmetric hda} for an example.
\paragraph{Terminology.}
Since our interest lies in symmetrisations, we use the following terminology.
For a precubical set $X$, we call $SX$ its \emph{symmetrised precubical set};
objects of this form will be referred to as \emph{sprecubical sets}.
For an HDA $\mathcal X=(X,i_X)$, we write $\mathcal X := (SX,\,(\id.i_X))$ for its
\emph{symmetric HDA} (or \emph{sHDA}).

\paragraph{Remark}
There also exists a functor
$s^{*}:\widehat{\Xi}\to\widehat{\square}$
obtained by precomposition with the inclusion
$s:\square\hookrightarrow\Xi$, that is, $s^{*}(X) := X\circ s.$

This functor keeps the same underlying sets of cells and the same
coface maps, but discards all symmetry morphisms of~$\Xi$.  
Consequently, it also forgets the identifications between a cell and
its permuted copies that were induced by those symmetries in the first place.
For a symmetric precubical set $Y$, the ordered precubical set $s^{*}Y$
thus contains all cubes of $Y$, but treats previously equivalent
permuted cells as independent ones.
For instance, if $Y=SX$ as in Example~\ref{ex:single-2d-cell},
then $Y$ contains a $2$–cell in shape $(2,ab)$ together with its
permuted companion in shape $(2,ba)$. After applying $s^{*}$, both
cells remain present, but the symmetry morphism between them is no
longer available; in $s^{*}Y$ they are interpreted as two independent
ordered squares.
Applying the symmetriser again to $s^{*}Y$ freely adds new symmetric
copies of each of these cells, so $y$ and $(12).y$ each generate their
own pair of permuted variants.
In general, every $n$–cell of $Y$ gives rise to $n!$ copies after the
first symmetrisation, $(n!)!$ copies after the second, and so on,
resulting in a factorial growth under iteration. 
\section{Pomsets}\label{sec:pomsets}
Partially ordered multisets (\emph{pomsets}) have long been used to model true concurrency and have played a central role in theoretical and applied concurrency research~\cite{GISCHER1988199,grabowski1981partial,Pratt1986ModelingCW}. Earlier work on higher-dimensional automata (HDAs) extended pomsets by adding an event order on incomparable events and by using interfaces to distinguish the beginning and end of a run~\cite{generatingposetbeyondn,LanguageofHDA}. In this paper we drop the event order: unlike~\cite{LanguageofHDA}, we describe the observable content of HDA paths using pomsets that carry only the precedence order, while following the interface construction and gluing composition of~\cite{generatingposetbeyondn}.
 \begin{definition}[Pomset]
    A \emph{partially ordered multiset (pomset)} is a tuple $(P, <_P,\lambda_P )$ where $P$ is a finite set, $\lambda_{P}: P \rightarrow \Sigma$ is a labeling function over an alphabet $\Sigma$, $<_{P}$ is a strict partial order on $P$ called \emph{precedence order}.
    \end{definition}
   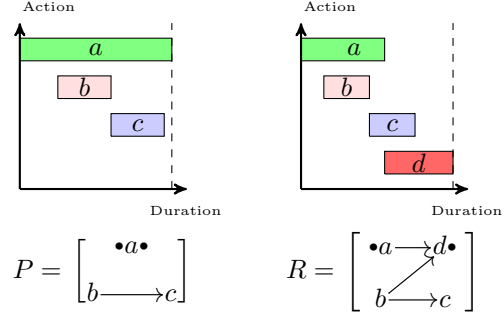
\begin{figure}%[tbp]
  \centering
  \begin{tikzpicture}[x=1cm]
    \def\possh{-1.3}
    \begin{scope}[shift={(0.0,0)}]
      \def\hw{0.3}
      \filldraw[fill=green!50!white,-](0,1.2)--(2,1.2)--(2,1.2+\hw)--(0,1.2+\hw);
      \filldraw[fill=pink!50!white,-](.5,0.7)--(1.2,0.7)--(1.2,0.7+\hw)--(.5,0.7+\hw)--(.5,0.7);
      \filldraw[fill=blue!20!white,-](1.2,0.2)--(1.9,0.2)--(1.9,0.2+\hw)--(1.2,0.2+\hw)--(1.2,0.2);
      \draw[thick,->](0,-.5)--(0,1.7);
      \draw[thick,->](0,-.5)--(2.2,-.5);
      \draw[thin,-,dashed](2,-.5)--(2,1.7);
       \node at (0.4,1.8) { $^{{}^{\text{Action}}}$};
       \node at (2.2,-.9) { $^{{}^{\text{Duration}}}$};
      \node at (1,1.2+\hw*0.5) {$a$};
      \node at (0.85,0.7+\hw*0.5) {$b$};
      \node at (1.55,0.2+\hw*0.5) {$c$};
      \node at (1.1,-1.8+\hw*0.7)  {$P=\bigipomset[.4]{ & \ibullet a \ibullet  \\ ~ b \ar[rr] && c ~ }$}; %{$ \!\!\ipomset{ & \ibullet a \ar@{.>}[d]  \ar[r]  &  b \\  & c  &}$}; %
    \end{scope}
   
    \begin{scope}[shift={(3.7,0)}]
       \def\hw{0.3}
       \filldraw[fill=green!50!white,-](0,1.2)--(1.1,1.2)--(1.1,1.2+\hw)--(0,1.2+\hw);
     \filldraw[fill=pink!50!white,-](.3,0.7)--(0.9,0.7)--(0.9,0.7+\hw)--(.3,0.7+\hw)--(.3,0.7);
     \filldraw[fill=blue!20!white,-](0.9,0.2)--(1.5,0.2)--(1.5,0.2+\hw)--(0.9,0.2+\hw)--(0.9,0.2);
       \filldraw[fill=red!60!white,-](1.1,-0.3)--(2,-0.3)--(2,-0.3+\hw)--(1.1,-0.3+\hw)--(1.1,-0.3);
      \draw[thick,->](0,-.5)--(0,1.7);
      \draw[thick,->](0,-.5)--(2.2,-.5);
      \draw[thin,-,dashed](2,-.5)--(2,1.7);
       \node at (0.4,1.8) { $^{{}^{\text{Action}}}$};
       \node at (2.2,-.9) { $^{{}^{\text{Duration}}}$};
      \node at (.7,1.2+\hw*0.5) {$a$};
     \node at (0.6,0.7+\hw*0.5) {$b$};
      \node at (1.2,0.2+\hw*0.5) {$c$};
      \node at (1.55,-0.3+\hw*0.5) {$d$};
        \node at (1.1,-1.8+\hw*0.7){$R=\bigipomset{& \ibullet a \ar[r] & d \ibullet  & \\ & b \ar[ur] \ar[r] & c & }$}; 
    \end{scope}
  \end{tikzpicture}
  \caption{Interval ipomsets (below) with their corresponding interval representations (above). An event with a dot on the left (resp.\@ on the right) is an element of a source (resp.\@ target) interface. Full arrows indicate precedence order.}
  \label{fi:interval rep}
\end{figure}
To model concurrency using pomsets, elements of a pomset $P$ represent \emph{events}, and $x <_P y$ signifies that event $x$ must occur before event $y$.  Two distinct events $x,y \in P$ are \emph{concurrent}, written $x \parallel y$, precisely when neither $x <_P y$ nor $y <_P x$.  An event $x \in P$ is \emph{minimal} if no event precedes it and \emph{maximal} if it precedes no event; we denote the sets of minimal and maximal events by $P_{\min}$ and $P_{\max}$, respectively. A subset \(Q\subseteq P\) is an \emph{antichain} if any two distinct elements of \(Q\) are concurrent. An antichain is \emph{maximal} if it is not properly contained in any larger antichain. In a pomset, antichains correspond precisely to concsets.   
\begin{definition}[ipomset ]
Let \(U=(m,\lambda)\) and \(V=(l,\gamma)\) be canonical concsets. A \emph{ pomset with interfaces (ipomset)} consists of a pomset \(P\) and two
\emph{injective, label-preserving} functions $U \xrightarrow{\;\si\;} P \;\xleftarrow{\;\ti\;} V,$ such that
\(\si(U) \subseteq P_{\min}\) called the source interface and
\(\ti(V) \subseteq P_{\max}\) called the target interface.
We denote it by $ \mathbf{P
}=(\si,\,P,\,\ti):\; U \to V$. 
\end{definition}
The source and target interfaces of an ipomset are antichains, hence concsets.
In this paper we work exclusively with interfaces whose domains are
\emph{canonical concsets} and whose maps are canonical, namely either the
identity or the standard injection
\(\iota_i\colon(n{-}1)\hookrightarrow(n)\).
Although concsets in general admit multiple label-preserving bijections
(see Example~\ref{exp: diff gluing results}), all interfaces
appearing here are uniquely determined, and no further choices arise.

As we drop the event order, we follow the gluing construction for pomsets with interfaces from \cite{generatingposetbeyondn}. This gluing operation is associative, admits a unit, but is not commutative; commutativity is neither expected nor needed here, since gluing models sequential concatenation of paths.
\begin{definition}[Gluing composition of ipomsets ]\label{def:ispomset-glue} Let
   $(s_1,P_1,t_1):$ $ (n,\nu)\to(m,\mu)$ and \((s_2,P_2,t_2)\colon (m,\mu)\to(k,\kappa)\) be ipomsets. Their \emph{gluing composition} is the ipomset $ (s_1,\;P_1 * P_2,\;t_2)\;:\;(n,\nu)\;\longrightarrow\;(k,\kappa),$ where the carrier is $P_1 * P_2 :=
   \bigl(\; (P_1\sqcup P_2)\,/\,{t_1(i)=s_2(i)},\; \le, \lambda_1\cup\lambda_2 \bigr),$ the disjoint union of \(P_1\) and \(P_2\)
        with the interface elements identified:
        \(t_1(i)\equiv s_2(i)\) for all \(i\in[m]\), the precedence order is $ \le\;=\;   \bigl(\,\le_{1}\;\cup\;\le_{2}\;\cup\;        (P_1\!\setminus\!t_1[m])\times    (P_2\!\setminus\!s_2[m])\bigr)$ and the labeling function is inherited component-wise:
        \(\lambda_1\) on \(P_1\) and \(\lambda_2\) on \(P_2\).    
We regard \(P_1\) and \(P_2\) as sub-pomsets of \(P_1*P_2\) via the canonical injections.
\end{definition}
Although different choices of interface identifications may in general lead to non-isomorphic gluing results (see Example~\ref{exp: diff gluing results}), such situations will not arise in this paper. Indeed, all ipomsets we consider have interfaces whose domain and codomain are canonical concsets, so the gluing operation is uniquely determined.
Moreover, every interface map $f$ is strictly order preserving in the sense that for all $i,j$, $i <_{\mathbb{N}} j \;\Longleftrightarrow\; f(i) <_{\mathbb{N}} f(j),$ where $<_{\mathbb{N}}$ denotes the usual strict order on natural numbers.
In particular, interface identifications preserve $<_{\mathbb{N}}$, ruling out any ambiguity in the gluing construction.
\begin{example}\label{exp: diff gluing results}
Changing the identification of interface events may lead to non-isomorphic gluing results.
For instance, consider the following two gluings of discrete ipomsets.

\noindent
\emph{Order-preserving identification.}
Identifying the interfaces via the order-preserving bijection
$t_1(i)=s_2(i)$ for $i=1,2$ yields
\[
\!\!\ipomset{& a \ar[r]  &  b=t_1(1)  \ibullet  &  \\&& b=t_1(2) \ibullet &}
\;*\;
\!\!\ipomset{ & \ibullet b=s_2(1)    \ar[r]  &c & \\& \ibullet b=s_2(2)  & &}
\;=\;
\!\!\ipomset{& a \ar[r]  &  b  \ar[r]  & c &\\&& b  & &}.
\]
This is the situation that arises throughout this paper.

\noindent
\emph{Order-reversing identification.}
If instead the interface is identified via the permutation
$t_1(1)=s_2(2)$ and $t_1(2)=s_2(1)$, one obtains
\[
\!\!\ipomset{& a \ar[r]  &  b=t_1(1)  \ibullet  &  \\&& b=t_1(2) \ibullet &}
\;*\;
\!\!\ipomset{ & \ibullet b=s_2(2)    \ar[r]  &c & \\& \ibullet b=s_2(1)  & &}
\;=\;
\!\!\ipomset{& a \ar[r] \ar[dr]  &  b   & \\&b \ar[r] & c  & }.
\]
Such a gluing is admissible in general for ipomsets, but it violates
order preservation of the interfaces and therefore does not occur under
the assumptions imposed in this paper.
\end{example}

\begin{definition}[Interval Ipomset \cite{fishburn1985interval}]
An \emph{interval ipomset} is an ipomset \(P\) such that for any \(x, y, z, w \in P\), if \(x <_P z\) and \(y <_P w\), then \(x <_P w\) or \(y <_P z\). In other words, it does not contain an induced subpomset of the form:
$\twotwo = \pomset{\bullet \ar[r] & \bullet \\ \bullet \ar[r] & \bullet}$.

\end{definition}
\begin{proposition}(\cite{LanguageofHDA,janicki1993structure})\label{prop:decomposition-into-discrete}
Let \(P\) be an ipomset.  
Then \(P\) is an interval ipomset if and only if \(P\) can be expressed as a
finite gluing of discrete ipomsets.
\end{proposition}
In this work, we focus solely on interval ipomsets: \textbf{all ipomsets are assumed to be interval even if not stated explicitly.}  
\begin{definition}[Isomorphism of ipomsets]\label{def:isomorphism-ispomsets}
Let \(\mathbf{P}=(s_P,P,t_P)\colon U\to V\) and
\(\mathbf{Q}=(s_Q,Q,t_Q)\colon U_Q\to V_Q\) be ipomsets.
An \emph{isomorphism} \(\mathbf{P}\cong\mathbf{Q}\) is a bijection
\(f\colon P\to Q\) such that $\lambda_P = \lambda_Q\circ f,
x<_P y \iff f(x)<_Q f(y),$
and \(f\) restricts to bijections \(f_U\colon U\to U_Q\) and
\(f_V\colon V\to V_Q\) satisfying
\(s_Q\circ f_U = f\circ s_P\) and \(t_Q\circ f_V = f\circ t_P\). Equivalently, the following diagram commutes:
$$\begin{tikzcd}[row sep=.5cm, column sep=1.8cm]
U \arrow[r, "s_P"] \arrow[d, "f_U" swap]
  & P \arrow[d, "f"]
  & V \arrow[l, "t_P" swap] \arrow[d, "f_V"] \\[0.2cm]
U_Q \arrow[r, "s_Q" swap]
  & Q
  & V_Q \arrow[l, "t_Q"]
\end{tikzcd}$$
\end{definition}
There is at most one isomorphism between two ipomsets in the presence of event order~\cite{LanguageofHDA}.
This uniqueness fails in our case.

\section{Paths in HDAs and sHDAs}\label{sec: path and their labels}
We recall the notion of paths in HDAs and sHDAs and the standard relations
between them, fixing notation for later use.
We also make explicit the relationship between paths and their symmetric
liftings.
\begin{definition}\label{def:path}   
 A \emph{path} of length $n$ in a precubical set $X$ is a sequence $
\alpha=\left(x_{0}, \varphi_{1}, x_{1}, \varphi_{2}, \ldots, \varphi_{n}, x_{n}\right),$
where $x_{j} \in X\left[U_{j}\right]$ are cells, and for all $j$, either
\begin{itemize}
  \item $\varphi_{j}=d_{i_{j}}^{0} \in \square\left(U_{j-1}, U_{j}\right)$ a source map and $x_{j-1}=\delta_{i_{j}}^{0}\left(x_{j}\right)$ (up-step), or
  \item $\varphi_{j}=d_{i_{j}}^{1} \in \square\left(U_{j}, U_{j-1}\right), \delta_{i_{j}}^{1}\left(x_{j-1}\right)=x_{j}$ (down-step).
\end{itemize}
A path in a sprecubical set $Y$ is a path in the underlying precubical set $s^*Y$.
\end{definition}
%===================================
\begin{lemma}[Paths in a symmetrised precubical set]\label{lem:paths-in-SX}
Let $X\in\widehat{\square}$ and consider its symmetrisation $SX\in\widehat{\Xi}$.
A \emph{path} of length $n$ in the sprecubical set $SX$ is equivalently a sequence
$$\alpha = \bigl((\tau_{0}.x_{0}),\,\varphi_{1},\,(\tau_{1}.x_{1}),\ldots,
\,\varphi_{n},\,(\tau_{n}.x_{n})\bigr),$$
where each $x_j\in X[U_j]$, each $\tau_j\in\mathfrak S_{|U_j| }$ is a permutation,
and each $\varphi_j=d^{k_j}_{p_j}$ is a coface map in $\Xi$, such that for every
$j=1,\ldots,n$:
\begin{itemize}
\item[(i)] (\emph{step condition}) either
\begin{itemize}
  \item $k_j=0$ and $(\tau_{j-1}.x_{j-1}) \;=\; (SX)[d^{0}_{p_j}](\tau_j.x_j),$ or
  \item $k_j=1$ and $(\tau_{j}.x_{j}) \;=\; (SX)[d^{1}_{p_j}](\tau_{j-1}.x_{j-1}).$
\end{itemize}
\item[(ii)] (\emph{permutation coherence}) the permutations satisfy
$$\tau_{j-1} = d_{p_j}\tau_j  \text{ if } k_j = 0, \qquad \text{ and }\qquad  \tau_{j}   = d_{p_j}\tau_{j-1}  \text{ if } k_j = 1.$$

\end{itemize}
\end{lemma}
\begin{proof}
By definition, a path in the sprecubical set $SX\in\widehat{\Xi}$ means a path
in the underlying precubical set $s^*(SX)\in\widehat{\square}$, i.e.\ a sequence
of cells related by face maps in the sense of Definition~\ref{def:path}. Unfolding the presheaf action of $SX$ on a coface map $d^k_{p}$, we have $$(SX)[d^k_{p}](\tau.x)
=
\bigl(d_{p}\tau.\;X[d^k_{\tau^{-1}(p)}](x)\bigr),$$
where $d_{p}\tau$ denotes the induced permutation on the corresponding face.
Hence, for an up-step ($k=0$) the condition
\(
(\tau_{-}.x_{-})=(SX)[d^0_{p}](\tau_{+}.x_{+})
\)
is equivalent to the pair of equalities $\tau_{-}=d_{p}\tau_{+}
$ and $
x_{-}=X[d^0_{\tau_{+}^{-1}(p)}](x_{+}),$
and for a down-step ($k=1$) the condition
\(
(\tau_{+}.x_{+})=(SX)[d^1_{p}](\tau_{-}.x_{-})
\)
is equivalent to $\tau_{+}=d_{p}\tau_{-}
$ and $
x_{+}=X[d^1_{\tau_{-}^{-1}(p)}](x_{-}).$
Applying these equivalences stepwise along the sequence yields exactly the
up-step/down-step clauses in~(i), and the permutation coherence equations in~(ii).
Conversely, if the equalities in~(i) (equivalently, the expanded conditions above)
hold at every step, then each consecutive triple forms a valid face step in
$s^*(SX)$, hence the whole sequence is a path in $SX$.
\end{proof}

\paragraph{Relation between paths in $X$ and in $SX$.}
Let $X$ be a precubical set and let $\alpha=(x_0, d^{k_1}_{i_1}, x_1, \ldots, d^{k_m}_{i_m}, x_m)\in\Path_X$ be a path in $X$.
A \emph{lifting} of $\alpha$ to the sprecubical set $SX$ is a path of the form
\[
\beta=
\bigl(
(\tau_0.x_0),\,
d^{k_1}_{p_1},\,
(\tau_1.x_1),\,
\ldots,\,
d^{k_m}_{p_m},\,
(\tau_m.x_m)
\bigr)\in\Path_{SX},
\]
where the permutations $\tau_0,\ldots,\tau_m$ and indices $p_1,\ldots,p_m$
satisfy
\[
p_j=\tau_{j-k_j}(i_j)
\quad\text{and}\quad
\begin{cases}
\tau_{j-1}=d_{p_j}\tau_j & \text{if } k_j=0,\\[1mm]
\tau_j=d_{p_j}\tau_{j-1} & \text{if } k_j=1,
\end{cases}\forall j=1,\ldots,m.
\]
Among all liftings of $\alpha$, the \emph{canonical lifting}
$$S\alpha=
\bigl(
(\id.x_0),\,d^{k_1}_{i_1},\,(\id.x_1),\,\ldots,\,
d^{k_m}_{i_m},\,(\id.x_m)
\bigr)$$
is obtained by choosing $\tau_j=\id$ for all $j$.

Conversely, Lemma~\ref{lem:paths-in-SX} shows that every path
$\beta\in\Path_{SX}$ is of this form for a unique underlying path
$\alpha\in\Path_X$, obtained by setting $i_j=\tau_{j-k_j}^{-1}(p_j)$. 
\paragraph{Path notation and operations.}
Let $\alpha=(x_{0},\varphi_{1},\ldots,\varphi_{n},x_{n})$ be a path in a
precubical set $X$ or in a sprecubical set $SX$.

A \emph{path in an HDA} $\mathcal{X}=(X,i_X)$ is a path in the
underlying precubical set $X$ whose first cell is the initial cell $i_X$.
Similarly, a path in the symmetric HDA $S\mathcal{X}$ is a path in $SX$
whose first cell is $(\id.i_X)$.
We write $\Path_X \text{ and } \Path_{SX}$ for the sets of all paths in $X$ and $SX$, respectively, and
$$\Path_{\mathcal{X}}\subseteq\Path_X
\quad\text{and}\quad
\Path_{S\mathcal{X}}\subseteq\Path_{SX}$$
for the corresponding sets of paths starting at the initial cell. This distinction matters for the semantics below.
The ipomset label are defined for \emph{all}
combinatorial paths in $\Path_X$ and $\Path_{SX}$, independently of the
initial cell.
By contrast, the ST--trace is defined only for executions, that is, for
paths in $\Path_{\mathcal{X}}$ and $\Path_{S\mathcal{X}}$.

If $\alpha=(x_{0},\varphi_{1},\ldots,\varphi_{n},x_{n})
\quad\text{and}\quad
\beta=(y_{0},\psi_{1},\ldots,\psi_{m},y_{m})$ are paths in $X$ or in $SX$ with $x_n=y_0$, their
\emph{concatenation} is the path
$$\alpha*\beta=
\bigl(x_{0},\varphi_{1},\ldots,\varphi_{n},x_{n},
\psi_{1},y_{1},\ldots,\psi_{m},y_{m}\bigr).$$
In this case, we say that $\alpha$ is a \emph{restriction} of $\alpha*\beta$.
\begin{example}[Paths and liftings in the symmetrised square]
\label{ex:path-liftings-picture}
Consider the HDA $\mathcal X=(X,i_X)$ depicted on the left of
Figure~\ref{fig:hda and its symmetric hda}, consisting of a single $2$--cell $x \in X(2,ab),$ together with its faces.
The symmetrised HDA $S\mathcal X$ is shown on the right of the figure and
contains two $2$--cells, $x \in (SX)(2,ab)
\text{ and }
x' \in (SX)(2,ba),$
with $\delta^k_1(x)=\delta^k_2(x')$ for $k=0,1$. Consider the path in $X$ \linebreak $\alpha
=
\bigl(
\delta^0_1(x),\;
d^0_2,\;
x,\;
d^1_1,\;
\delta^1_2(x)
\bigr),$ which enters the square $x$ along the $b$--edge and exits it along the
$a$--edge. In the symmetrised precubical set $SX$, the path $\alpha$ admits exactly
two liftings.
The canonical lifting passes through the cell $x\in(SX)([2,ab])$ and is
given by $S\alpha
=
\bigl(
\delta^0_1(x),\;
d^0_2,\;
x,\;
d^1_1,\;
\delta^1_2(x)
\bigr).$ The second lifting passes through the symmetric copy
$x'\in(SX)(2,ba)$.
Since the event order is reversed, the face indices are exchanged,
yielding the path $\alpha'
=
\bigl(
\delta^0_2(x'),\;
d^0_1,\;
x',\;
d^1_2,\;
\delta^1_1(x')
\bigr).$ %Thus the two liftings of $\alpha$ correspond exactly to the two possible
%event orders $(2,ab)$ and $(2,ba)$ represented in the symmetrised HDA.
\end{example}

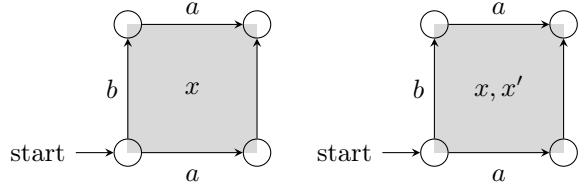
\begin{figure}[h!]
  \centering
  \setlength{\tabcolsep}{.2cm}
  \begin{tabular}{ccc}
  
  \begin{tikzpicture}[x=1.7cm, y=1.7cm, >=stealth]
  %--- node (state) style -------------------------------------------------------
  
       \path[fill=black!15] (0,0) to (1,0) to (1,1) to (0,1);
      \node[state, initial, minimum size=8pt] (00) at (0,0) {};
      \node[state, minimum size=8pt] (10) at (1,0) {};
      %\node[state] (20) at (2,0) {};
      \node[state, minimum size=8pt] (01) at (0,1) {};
      \node[state, minimum size=8pt] (11) at (1,1) {};
    %  \node[state] (21) at (2,1) {};
      %\node[state, accepting] (21) at (2,1) {};
      \path (00) edge node[below] {$\vphantom{d}a$} (10);
     % \path (11) edge node[above] {$\vphantom{d}b$} (21);
      %\path (10) edge node[below] {$d$} (20);
      \path (01) edge node[above] {$a$} (11);
      %\path (11) edge node[above] {$d$} (21);
      \path (00) edge node[left] {$b$} (01);
      \path (10) edge node[right] {} (11);
      \node at ($(0,0)!0.5!(1,1)$) {\(x\)};
\end{tikzpicture}
    
    &
  \begin{tikzpicture}[x=1.7cm, y=1.7cm, >=stealth]
  %--- node (state) style -------------------------------------------------------
  
       \path[fill=black!15] (0,0) to (1,0) to (1,1) to (0,1);
      \node[state, initial, minimum size=8pt] (00) at (0,0) {};
      \node[state, minimum size=8pt] (10) at (1,0) {};
      %\node[state] (20) at (2,0) {};
      \node[state, minimum size=8pt] (01) at (0,1) {};
      \node[state, minimum size=8pt] (11) at (1,1) {};
    %  \node[state] (21) at (2,1) {};
      %\node[state, accepting] (21) at (2,1) {};
      \path (00) edge node[below] {$\vphantom{d}a$} (10);
   %   \path (11) edge node[above] {$\vphantom{d}b$} (21);
      %\path (10) edge node[below] {$d$} (20);
      \path (01) edge node[above] {$a$} (11);
      %\path (11) edge node[above] {$d$} (21);
      \path (00) edge node[left] {$b$} (01);
      \path (10) edge node[right] {} (11);
      \node at ($(0,0)!0.5!(1,1)$) {\(x,x'\)};
     % \draw[->, very thick, blue] (0,0) --  (.5,0)--  (.5,0.5);
 %     \path (10) edge (11);
      %\path (20) edge node[right] {$a$} (21);
        %\draw[-, very thick, red] (0,0) -- (1,0) -- (1,.5) -- (1.5,.5) -- (1.5,1)  -- (2,1);
% \node at (.3,.6) {$\vphantom{b} \mathbf{y}_U $};
%   \node[font=\footnotesize] at (1.4,.2) {\textcolor{orange}{}};
    %   \draw[-,  very thick, orange] (0,0) --  (2,0) ;
    %   \node[] at (1.7,0.2) {\textcolor{orange}{$\alpha_2$}};
       %  \node[state, fill=orange, minimum size=0.2mm] at (1,.5) {};
         %\node[font=\footnotesize] at (.3,.2) {\textcolor{blue}{}};
       %    \draw[-, very thick, blue] (0,0) -- (1,1) ;
  
\end{tikzpicture}
    
  \end{tabular}
  \bigskip
  \caption{HDA $\mathcal{X}$ on the left with its symmetriser $S\mathcal{X}$ on the right. }
  \label{fig:hda and its symmetric hda}
\end{figure}
The following is an example in 3 dimensions. It illustrates how the symmetriser affects paths and how
permutations must act coherently along a path in~$SX$.

\begin{example}\label{ex:path_symmetrization}
Consider the precubical set \(X\) represented in Figure~\ref{fig:3D HDA}, where \(x,z\in X[3]\) the two 3-cells sharing the 2-cell \(y\), and consider $\alpha=(x,\ d^{1}_{3},\ y,\ d^{0}_{3},\ z)\in \Path_X,$ which first terminates the event \(c\) and then initiates the event \(d\). There are $3!=6$ distinct liftings of $\alpha$ in $SX$. These include:  
%Applying the symmetriser yields the canonical symmetric copy \linebreak $S\alpha=((\mathrm{id}\!\cdot x),\, d^{1}_{3},\, (\mathrm{id}\!\cdot y),\, d^{0}_{3},\, (\mathrm{id}\!\cdot z)),$ and also other symmetric copies obtained by applying permutations.
%In a symmetric HDA, each face step must satisfy the compatibility condition
%\(\theta_{-}=d_{i}\theta_{+}\), so a permutation must act
%\emph{uniformly} along the entire path.
%Hence, there is no valid path in \(SX\) that goes from
%\((\mathrm{id}\!\cdot x)\) to \((\mathrm{id}\!\cdot z)\) through
%\((12)\!\cdot y\); the intermediate 2-cell \((12)\!\cdot y\) forces the
%endpoints to be \((12)\!\cdot x\) and \((12)\!\cdot z\).
%The following symmetric copies of \(\alpha\) are induced:
\begin{enumerate}
    \item $ \beta
      =\bigl((12)\!\cdot x,\ d^{1}_{3},\ (12)\!\cdot y,\ d^{0}_{3},\ (12)\!\cdot z\bigr),$ the symmetric lifting corresponding to the transposition \((12)\) acting on $\ev(x)$.
   % and every step satisfies the required symmetry--face compatibility \(\theta_- = d_3\theta_+ = (12)\).
    \item $ \beta'
      =\bigl(\sigma\!\cdot x,\ d^{1}_{1},\ (12)\!\cdot y,\ d^{0}_{1},\ \sigma\!\cdot z\bigr),$ where \(\sigma=(123)\), the reduced permutation is
    \(d_{1}\sigma=(12)\). Here the intermediate 2-cell is \((d_1\sigma)\!\cdot y=(12)\!\cdot y\), and the face
    maps \(d^{1}_{1}\) and \(d^{0}_{1}\) correspond to \(d^{1}_{3}\) and \(d^{0}_{3}\) on the underlying path \(\alpha\), as enforced by the equivariance squares (via the permutation $\sigma$ and the equation $\theta_- = d_i\theta_+$).
\end{enumerate}
\end{example}
We recall the notion of adjacency of paths, as introduced by
van Glabbeek \cite{VANGLABBEEK2006265}. Adjacency and the derived notion of congruence apply uniformly to paths in a
precubical set $X$ and in its symmetric counterpart $SX$ \cite{KAHL202247}.
\begin{definition}[Adjacency of paths]\label{def:adjacency} Two paths $\alpha$ and $\alpha'$ are said to be \emph{adjacent},
written $\alpha \leftrightsquigarrow \alpha'$, if one can be obtained
from the other by a single local replacement of one of the following
forms, for indices $i<j$:

\begin{minipage}[t]{0.48\linewidth}
\begin{enumerate}
  \setlength{\itemsep}{0pt}
  \setlength{\parskip}{0pt}
  \setlength{\parsep}{0pt}
  \item[(1)] $(d_i^0,\;x_\ell,\;d_j^0)\;\rightsquigarrow\;(d^0_{j-1},\;x'_\ell,\;d_i^0)$,
  \item[(2)] $(d_j^1,\;x_\ell,\;d_i^1)\;\rightsquigarrow\;(d_i^1,\;x'_\ell,\;d_{j-1}^1)$,
\end{enumerate}
\end{minipage}\hfill
\begin{minipage}[t]{0.48\linewidth}
\begin{enumerate}
  \setlength{\itemsep}{0pt}
  \setlength{\parskip}{0pt}
  \setlength{\parsep}{0pt}
  \item[(3)] $(d_i^0,\;x_\ell,\;d_j^1)\;\rightsquigarrow\;(d^1_{j-1},\;x'_\ell,\;d_i^0)$,
  \item[(4)] $(d_j^0,\;x_\ell,\;d_i^1)\;\rightsquigarrow\;(d_i^1,\;x'_\ell,\;d_{j-1}^0)$.
\end{enumerate}
\end{minipage}
%Adjacency is a \emph{symmetric} relation: if $\alpha \leftrightsquigarrow \alpha'$
%then $\alpha' \leftrightsquigarrow \alpha$.
\end{definition}

For any $\alpha$ of length $m$ and $1\le\ell<m$, there exists a unique
path $\alpha^{(\ell)}$ obtained from $\alpha$ by applying the
corresponding adjacency replacement to the segment
$(\varphi_\ell, x_\ell, \varphi_{\ell+1})$ of $\alpha$
\cite{VANGLABBEEK2006265}.
We write $\alpha \overset{\ell}{\leftrightsquigarrow} \alpha^{(\ell)}$.

Adjacency rules (1) and (2) exchange two faces of the same polarity and are
reversible, whereas rules (3) and (4) interchange a start with a
termination and are directed.
The reversible rules generate a symmetric equivalence relation on paths,
called \emph{congruence}, which is the notion relevant for our purposes.
The directed rules induce subsumption; see \cite{KleeneTh}.
\begin{definition}[Congruence of paths]\label{def:path-congruence}
\emph{Congruence} $\simeq$ is the relation on paths generated by, for indices $i<j$:
\begin{itemize}
        \item $(d_i^0,\;x_\ell,\;d_j^0)\;\simeq\;(d^0_{j-1},\;x'_\ell,\;d_i^0)$;
        \item $(d_j^1,\;x_\ell,\;d_i^1)\;\simeq\;(d_i^1,\;x'_\ell,\;d_{j-1}^1)$;
      \item $\gamma * \alpha * \delta \simeq \gamma * \beta * \delta$ whenever $\alpha \simeq \beta$.
    \end{itemize}
    where in each case $x'_\ell$ is the unique cell determined by the
adjacency replacement in Definition~\ref{def:adjacency}.
\end{definition}
\section{Observable content}\label{sec: obs content}
This section collects the notions used to describe the observable behaviour
of paths. We recall the ST--trace semantics and introduce an event order free pomset-based notion of
path labels. This formulation and its extension to sHDAs are new. We then study how observable content behaves under symmetric liftings.
\subsection{ST-trace}
Fix a down--step $(x_\ell,d^1_{i_{\ell+1}},x_{\ell+1})$ occurring in a path $$\alpha=(x_0,\varphi_1,x_1,\ldots,\varphi_\ell,x_\ell,d^1_{i_{\ell+1}},x_{\ell+1},
\varphi_{\ell+2},\ldots,x_m)$$
in a (s)HDA.
By \cite{VANGLABBEEK2006265,KAHL202247}, there exists a unique index
$k\le \ell$ such that $$\alpha \overset{\ell}{\leftrightsquigarrow} \alpha^{(\ell)} 
\overset{\ell-1}{\leftrightsquigarrow} \alpha^{(\ell-1)} 
\overset{\ell-2}{\leftrightsquigarrow} \cdots 
\overset{k+1}{\leftrightsquigarrow} \alpha^{(k+1)} 
\overset{k}{\not\leftrightsquigarrow} \alpha^{(k)},$$

where each step is obtained by applying the unique adjacency replacement
at the indicated position.

This means that the down--step can be moved successively towards the
beginning of the path by adjacency replacements in the sense of
Definition~\ref{def:adjacency}, until it sits immediately after
its matching start step and cannot be moved further.

More explicitly, the down--step $d^1_{i_{\ell+1}}$ can be commuted one
position to the left precisely when the adjacent segment
$(\varphi_\ell,x_\ell,d^1_{i_{\ell+1}})$ matches one of the adjacency
patterns of Definition~\ref{def:adjacency}, that is, when the
indices of the two consecutive face maps satisfy the corresponding
inequality condition $i<j$ required there.

The process stops exactly at position $k$ because this condition fails:
in $\alpha^{(k+1)}$ the adjacent segment has the form
$(d^0_{i_k},x_{k+1},d^1_{i_k})$, where the indices coincide and no
adjacency replacement is applicable.
We then write $\mathsf{start}(i_{\ell+1})=k$.

\begin{definition}[ST–trace \cite{VANGLABBEEK2006265}]\label{def:split-ST-trace}
Let 
$
\alpha=(x_{0},
d^{\,k_{1}}_{i_{1}},x_{1},
\dots,
d^{\,k_{n}}_{i_{n}},x_{n})\in \Path_{S\mathcal{X}}$, and let $\lambda(i_j)$ denote the label of the
event whose start or termination is represented by the face map
$d^{k_j}_{i_j}$.
For each step define
$$\sigma_j^{ST} :=
\begin{cases}
  + & \text{if } k_j = 0,\\
  start(i_j) & \text{if } k_j = 1.
\end{cases}$$

Then $\textit{ST–trace}(\alpha)= (\lambda(i_1)^{\sigma^{ST}_1},\ldots,\lambda(i_n)^{\sigma^{ST}_n}).$
\end{definition}
The \emph{ST–trace} records not only when each action starts ($+$) and
terminates, but also links each termination to the position of its
corresponding start via the index $start(i_j)$.  In this way, the
ST–trace encodes the causal pairing between the beginning and ending of
individual actions, capturing their overlap and nesting
(see~\cite{VANGLABBEEK2006265} for the original formulation).

\medskip
\subsection{Observable content as pomsets}
Recall that for a cell $x$ in a (s)precubical set, $\ev(x)$ is a canonical
conclist (concset).
In this section, we systematically forget the event order and regard
$\ev(x)$ only as its underlying concset.
We now extend $\ev$ from cells to paths.
\begin{definition}\label{def: path pomset}[label of a path]
    Let $X$ be a precubical set and let $\alpha \in \Path_{X}$. The label of $\alpha$ is the ipomset $\ev(\alpha)$, computed recursively
 \begin{enumerate}
  \item \label{en: one cell} If $\alpha=(x)$ has length 0, then we set $\ev(\alpha)=(\mathrm{id}_{\ev(x)},\ev(x),\mathrm{id}_{\ev(x)}): \ev(x) \rightarrow \ev(x);$% $\ev(\alpha)= \ev(x) \stackrel{id}{\rightarrow} \ev(x)\stackrel{id}{\leftarrow}\ev(x);$
  \item %If $\alpha=\left(y \nearrow^{A}  x \right)$ or equivalently or equivalently $\delta^0_i(x)=y$, where $A \subseteq \ev(x)$ 
  If $\alpha=\left(y ,d_i^0, x \right)$, then $\ev(\alpha)$ is 
$(\iota_i,\ev(x),\mathrm{id}_{\ev(x)}): \ev(y) \rightarrow \ev(x);$ 
  \item \label{en: terminator } 
  If  $\alpha=\left(x , d_i^1, y \right)$, then $\ev(\alpha)$ is $(\mathrm{id}_{\ev(x)},\ev(x),\iota_i): \ev(x) \rightarrow \ev(y);$ %$\ev(x) \stackrel{id}{\rightarrow} \ev(x) \stackrel{\iota_{i}}{\leftarrow} \ev(y) ,$
  \item \label{en : label gluing} If $\alpha=\beta_{1} * \cdots * \beta_{m}$ is a concatenation of steps $\beta_{i}$, then $\ev(\alpha)=\ev\left(\beta_{1}\right) * \cdots * \ev\left(\beta_{m}\right).$
  % where $\iota_{i}:(n-1) \hookrightarrow (n)$ is the order-preserving
    injection skipping $i$.
\end{enumerate}
 Since paths in a symmetric precubical set $SX$ are formally paths in the
underlying precubical set $s^{*}(SX)$, the interval--pomset label
$\ev(\beta)$ is also well defined for every $\beta\in\Path_{SX}$.
\end{definition}
By construction, ipomset labels are built stepwise along paths; in particular, the label of any prefix embeds canonically into the label of the full path.

The label of any path $\alpha$ in a (s)precubical set is a finite gluing of discrete ipomsets. Consequently, by Proposition \ref{prop:decomposition-into-discrete}, $\ev(\alpha)$ is an interval ipomset for any path $\alpha$. This is why we restricted our focus on only interval ipomsets.
\paragraph{Comparison with ordered labels.}
In the literature, interval--pomset labels of paths are typically defined
using conclists, thereby retaining an explicit order on events; this
ordering is used to control the gluing operation
\cite{LanguageofHDA}.
In contrast, we systematically forget event order and work with the
underlying concsets.
The gluing operation remains well defined in our setting because all
interfaces involved are canonical, and interface maps are uniquely
determined.

Our formulation allows us to identify behaviours that differ only
by a permutation of concurrent events.
In particular, the interval pomsets $(a \parallel b)$ and
$(b \parallel a)$ are isomorphic in our setting, whereas they are
distinguished in ordered approaches, where isomorphisms must preserve event
order. This resolves a mismatch between pomset semantics and the intended treatment
of concurrency that appears in parts of the literature.

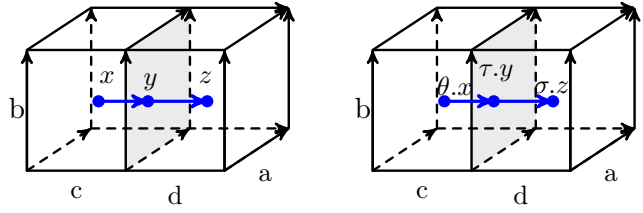
\begin{figure}
    \centering
     \begin{tikzpicture}[
  x={(1*1.3cm,0cm)},                 % projection of X axis
  y={(0.65*1.3cm,0.35*1.6cm)},           % projection of Y axis (depth)
  z={(0cm,1*1.6cm)},                 % projection of Z axis (height)
  line cap=round,line join=round,
  solid/.style={line width=0.9pt},
  hidden/.style={line width=0.9pt,dashed},
  face/.style={fill=black!8,draw=none},
  ghost/.style={draw=black!50,line width=0.6pt,opacity=0.7}
]
\def\hw{0.2}
% Box size: 2 (x) by 1 (y) by 1 (z)  ==> two attached unit cubes
\begin{scope}

    % --- 3D coordinates (front y=0; back y=1) ---
\coordinate (A) at (0,0,0);   % front-bottom-left
\coordinate (B) at (2,0,0);   % front-bottom-right
\coordinate (C) at (2,0,1);   % front-top-right
\coordinate (D) at (0,0,1);   % front-top-left

\coordinate (E) at (0,1,0);   % back-bottom-left
\coordinate (F) at (2,1,0);   % back-bottom-right
\coordinate (G) at (2,1,1);   % back-top-right
\coordinate (H) at (0,1,1);   % back-top-left

% Midpoints used for the shared-face and dots
\coordinate (Mfront) at (1,0,0);     % middle front bottom
\coordinate (MfrontTop) at (1,0,1);  % middle front top

%_________________
 \fill[face] (Mfront) -- (MfrontTop) -- (1,1,1) -- (1,1,0) -- cycle;
% ---------- Hidden edges (draw first) ----------
\draw[hidden] (A) -- (E);      % left bottom depth
\draw[hidden] (Mfront) -- ++(0,1,0); % middle bottom depth
\draw[hidden] (B) -- (F);      % right bottom depth
\draw[hidden] (E) -- (F);    
% back verticals that are hidden
\draw[hidden] (0,1,0) -- (0,1,1);     % (E)--(H)
\draw[hidden] (1,1,0) -- (1,1,1);     % interior back vertical
\draw[solid] (MfrontTop) -- (1,1,1);  %middle uper line
% ---------- Visible edges (outline) ----------
% front rectangle
\draw[solid] (A)--(B)--(C)--(D)--cycle;

% shared front vertical between the two cubes
\draw[solid] (Mfront) -- (MfrontTop);

% top rim (back edge visible)
\draw[solid] (D) -- (H) -- (G) ;
\draw[solid] (C) -- (G) ;
% right face outline
\draw[solid] (B) -- (F) -- (G);  % note: part of (B)--(F) hidden is already dashed; the
                                 % solid stroke on top produces the visible right edge

% left front vertical
\draw[solid] (A) -- (D);

% ---------- Labels “1” ----------
\node[scale=1.1] at (-0.1,0,.52) {b};       % left vertical
\node[below=2pt] at (0.5,0,0) {c};          % first cube length
\node[below=2pt] at (1.5,0,0) {d};   
\node[below=2pt] at (2.15,.4,0) {a}; 
\draw[solid] (B) -- (C) ;
%\node[below=2pt] at (1.5,0,0) {1};  
%%%%% The paths elements
\node[below=2pt] at (1.81,0,.75+\hw) {$z$};  
\node[below=2pt] at (.81,0,.75+\hw) {$x$};  
%\node[below=2pt] at (1.25,0,.91+\hw) {$y$};  
\node[below=2pt] at (1.25,0,.73+\hw) {$y$}; 
  \draw[->, very thick, blue] (0.4,0.5,.4) -- (.9,0.5,.4);
   \draw[->, very thick, blue] (.9,0.5,.4) -- (1.5,0.5,.4);
  \fill[blue] (.9,0.5,.4)circle[radius=2.3pt];
   \fill[blue] (1.5,0.5,.4)circle[radius=2.3pt];
    \fill[blue] (0.4,0.5,.4)circle[radius=2.3pt];
  
\end{scope}
\begin{scope}[shift={(3.5,0,0)}]
    % --- 3D coordinates (front y=0; back y=1) ---
\coordinate (A) at (0,0,0);   % front-bottom-left
\coordinate (B) at (2,0,0);   % front-bottom-right
\coordinate (C) at (2,0,1);   % front-top-right
\coordinate (D) at (0,0,1);   % front-top-left

\coordinate (E) at (0,1,0);   % back-bottom-left
\coordinate (F) at (2,1,0);   % back-bottom-right
\coordinate (G) at (2,1,1);   % back-top-right
\coordinate (H) at (0,1,1);   % back-top-left

% Midpoints used for the shared-face and dots
\coordinate (Mfront) at (1,0,0);     % middle front bottom
\coordinate (MfrontTop) at (1,0,1);  % middle front top

%_________________
 \fill[face] (Mfront) -- (MfrontTop) -- (1,1,1) -- (1,1,0) -- cycle;
% ---------- Hidden edges (draw first) ----------
\draw[hidden] (A) -- (E);      % left bottom depth
\draw[hidden] (Mfront) -- ++(0,1,0); % middle bottom depth
\draw[hidden] (B) -- (F);      % right bottom depth
\draw[hidden] (E) -- (F);    
% back verticals that are hidden
\draw[hidden] (0,1,0) -- (0,1,1);     % (E)--(H)
\draw[hidden] (1,1,0) -- (1,1,1);     % interior back vertical
\draw[solid] (MfrontTop) -- (1,1,1);  %middle uper line
% ---------- Visible edges (outline) ----------
% front rectangle
\draw[solid] (A)--(B)--(C)--(D)--cycle;

% shared front vertical between the two cubes
\draw[solid] (Mfront) -- (MfrontTop);

% top rim (back edge visible)
\draw[solid] (D) -- (H) -- (G) ;
\draw[solid] (C) -- (G) ;
\draw[solid] (B) -- (C) ;
% right face outline
\draw[solid] (B) -- (F) -- (G);  % note: part of (B)--(F) hidden is already dashed; the
                                 % solid stroke on top produces the visible right edge

% left front vertical
\draw[solid] (A) -- (D);

% ---------- Labels “1” ----------
\node[scale=1.1] at (-0.1,0,.52) {b};       % left vertical
\node[below=2pt] at (0.5,0,0) {c};          % first cube length
\node[below=2pt] at (1.5,0,0) {d};   
\node[below=2pt] at (2.15,.4,0) {a}; 
\node[below=2pt] at (1.81,0,.67+\hw) {$\sigma.z$};  
\node[below=2pt] at (.83,0,.71+\hw) {$\theta.x$};  
\node[below=2pt] at (1.25,0,.81+\hw) {$\tau.y$};  
%\node[below=2pt] at (1.25,0,.71+\hw) {$((12).y)$}; 
  \draw[->, very thick, blue] (0.4,0.5,.4) -- (.9,0.5,.4);
   \draw[->, very thick, blue] (.9,0.5,.4) -- (1.5,0.5,.4);
  \fill[blue] (.9,0.5,.4)circle[radius=2.3pt];
   \fill[blue] (1.5,0.5,.4)circle[radius=2.3pt];
    \fill[blue] (0.4,0.5,.4)circle[radius=2.3pt];
  
  %\fill[blue] (.9,0.5,.4) circle[radius=1.9pt];
\end{scope}
        % second cube length
 %\fill[face] (M0) -- (N0) -- (N1) -- (M1) -- cycle;
% ---------- Red points ----------
  % midpoint of top front edge of right cube
%\fill[red] (2,0,0.5) circle[radius=1.6pt];  % midpoint of right front vertical edge
% \caption{Paths and their labels in two-dimensional HDAs }
  %\label{fig:3D HDA}
\end{tikzpicture}
    \caption{On the left hand side, a precubical set \(X\) consisting of two 3-dimensional cubes \(x \in X(3,abc)\) and \(z \in X(3,abd)\) attached along a common 2-dimensional cell $y\in X(2,ab)$.
On the right hand side, Its symmetrization \(SX\), where each cube generates six 3-cells of the form
\(\theta\!\cdot\!x\) and \(\sigma\!\cdot\!z\) for some \(\theta,\sigma\in \SG_3\),
all sharing the same boundary faces $\tau.y$ and $\id.y$, where $\tau \in \SG_2$.}
    \label{fig:3D HDA}
\end{figure}
%\end{definition}
%\begin{enumerate}
 % \item \label{sen: one cell} If $\alpha=(x)$ has length 0, then $\ev(\alpha)= \ev(x) \stackrel{id}{\rightarrow} \ev(x)\stackrel{id}{\leftarrow}\ev(x)$.
 % \item If $\alpha=\left(y ,d_i^0, x \right)$, then $\ev(\alpha)= \ev(y) \stackrel{\iota_{i}}{\rightarrow} \tau(\ev(x)) \stackrel{id}{\leftarrow} \tau(\ev(x)) $ for some $\tau \in \SG$.\footnote{This I do not understand.}
 % \item \label{sen: terminator } If $\alpha=\left(x ,d_i^1,  y \right)$, then $\ev(\alpha)=\tau(\ev(x)) \stackrel{id}{\rightarrow} \tau(\ev(x)) \stackrel{\iota_{i}}{\leftarrow} \ev(y) $ for some $\tau \in \SG$.
 % \item \label{sen : label gluing} If $\alpha=\beta_{1} * \cdots * \beta_{n}$ is a concatenation of steps $\beta_{i}$, then $\ev(\alpha)=\ev\left(\beta_{1}\right) * \cdots * \ev\left(\beta_{n}\right).$
%\end{enumerate}
%Note that the first three cases give rise to discrete ipomsets as labels. In particular, the case of an up step $\left(y \nearrow^{A}  x \right)$ (resp. down step $(x \searrow_B y)$) gives rise to a starter of events $A$ (resp. a terminator of events $B$). %Although the observable content of a path has been expressed as ipomsets in \cite{LanguageofHDA}, our work extends this idea by defining the observable content of a path in a symmetrized HDA as an ispomset.
\begin{example}
    Consider the HDA $\mathcal{X}$ and its sHDA $S\mathcal{X}$ of Figure \ref{fig:3D HDA}.
    \begin{itemize}
        \item The label of $\alpha=(x,d^{1}_{3},y,d^{0}_{3},z)$ is the ipomset
        \begin{equation*}
        \begin{aligned}
        \ev(\alpha)
          &= (\mathrm{id}_{\ev(x)},\,\ev(x),\,\iota_{3})\mathrel{*}(\iota_{3},\,\ev(z),\,\mathrm{id}_{\ev(z)}):\ev(x)\to\ev(z)\\
          &= \left(\vcenter{\hbox{%
          \begin{tikzcd}[row sep=0.05cm, column sep=0.21cm, cells={nodes={inner sep=.5pt}}]
               \ibullet a \ibullet \\
               \ibullet b \ibullet \\
               \ibullet \vphantom{d}c \phantom{\ibullet}
          \end{tikzcd}}}\right)
          \mathrel{*}
          \left(\vcenter{\hbox{%
          \begin{tikzcd}[row sep=0.05cm, column sep=0.21cm, cells={nodes={inner sep=.5pt}}]
               \ibullet a \ibullet \\
               \ibullet b \ibullet \\
               \phantom{\ibullet} d \ibullet
          \end{tikzcd}}}\right)
          \\
          &= \left(\vcenter{\hbox{%
          \begin{tikzcd}[row sep=0.09cm, column sep=0.21cm, cells={nodes={inner sep=.5pt}}]
              & \ibullet a \ibullet & \\
              & \ibullet b \ibullet & \\
              \ibullet c\ar[rr] & & d \ibullet
          \end{tikzcd}}}\right),
        \end{aligned}
        \end{equation*}
        which corresponds to first terminating the event~$c$ and then starting~$d$ through the common face labelled~$ab$.
        \item The label of $\beta = \bigl((\theta\!\cdot x),\; d^{1}_{3},\; (d_3\theta\!\cdot y),\; d^{0}_{3},\; (\theta\!\cdot z)\bigr)$, where $\theta=(12)$, is
        \begin{equation*}
        \begin{aligned}
        \ev(\beta)
          &= \left(\vcenter{\hbox{%
          \begin{tikzcd}[row sep=0.05cm, column sep=0.21cm, cells={nodes={inner sep=.5pt}}]
               \ibullet b \ibullet \\
               \ibullet a \ibullet \\
               \ibullet \vphantom{d}c \phantom{\ibullet}
          \end{tikzcd}}}\right)
          \mathrel{*}
          \left(\vcenter{\hbox{%
          \begin{tikzcd}[row sep=0.05cm, column sep=0.21cm, cells={nodes={inner sep=.5pt}}]
               \ibullet b \ibullet \\
               \ibullet a \ibullet \\
               \phantom{\ibullet} d \ibullet
          \end{tikzcd}}}\right)\\
          &= \left(\vcenter{\hbox{%
          \begin{tikzcd}[row sep=0.09cm, column sep=0.21cm, cells={nodes={inner sep=.5pt}}]
              & \ibullet b \ibullet & \\
              & \ibullet a \ibullet & \\
              \ibullet c\ar[rr] & & d \ibullet
          \end{tikzcd}}}\right).
        \end{aligned}
        \end{equation*}

        \item The label of
        $\beta' = \bigl(\sigma\!\cdot x,\ d^{1}_{1},\ (12)\!\cdot y,\ d^{0}_{1},\ \sigma\!\cdot z\bigr)$, where $\sigma=(123)$, is
        \begin{equation*}
        \begin{aligned}
        \ev(\beta')
          &= \left(\vcenter{\hbox{%
          \begin{tikzcd}[row sep=0.05cm, column sep=0.21cm, cells={nodes={inner sep=.5pt}}]
               \ibullet \vphantom{d}c \phantom{\ibullet} \\
               \ibullet a \ibullet \\
               \ibullet b \ibullet
          \end{tikzcd}}}\right)
          \mathrel{*}
          \left(\vcenter{\hbox{%
          \begin{tikzcd}[row sep=0.05cm, column sep=0.21cm, cells={nodes={inner sep=.5pt}}]
               \phantom{\ibullet} d \ibullet \\
               \ibullet a \ibullet \\
               \ibullet b \ibullet
          \end{tikzcd}}}\right)\\
          &= \left(\vcenter{\hbox{%
          \begin{tikzcd}[row sep=0.09cm, column sep=0.21cm, cells={nodes={inner sep=.5pt}}]
              \ibullet c\ar[rr] & & d \ibullet \\
              & \ibullet a \ibullet & \\
              & \ibullet b \ibullet &
          \end{tikzcd}}}\right).
        \end{aligned}
        \end{equation*}
    \end{itemize}

\end{example}
%\begin{remark}[Autoconcurrency]\label{rem:autoconcurrency}
\paragraph{Remark on autoconcurrency.}
All constructions above apply equally in the presence of autoconcurrency.
Distinct concurrent events carrying the same label remain distinguished by
their structural positions in the underlying concsets.
Hence, path formation and interval--pomset labeling are unaffected by
autoconcurrency.
%\end{remark}
%The next lemma captures the two fundamental behaviours of labels under
%adjacency classes: subsumption steps strictly increase causal latitude
%(monotonicity), whereas same-polarity swaps preserve labels (invariance).
\subsection{Compatibility of symmetrisation with Observable content.}
We relate the observable content of a path in a precubical set $X$ to that of
its liftings to the symmetric precubical set $SX$.
In the following, we consider separately the behaviour of ST--traces and of
interval--pomset labels under symmetrisation.
\paragraph{ST--semantics.}
Lifting a path to the symmetric setting preserves its ST--trace.
In particular, the matching between start and termination events is
invariant under symmetrisation.
This follows from Theorem~6.1 of~\cite{KAHL202247}, which shows that every
HDA is hhp-bisimilar to its symmetric counterpart; hhp-bisimulation will
be introduced in Section~\ref{sec: Bis for HDA}.
\begin{proposition}
\label{lem:ST-invariant-under-S}
Let $\mathcal{X}=(X,i_X,F_X)$ be an HDA, and let $\alpha\in\Path_{\mathcal{X}}$.
If $\alpha'$ is a lifting of $\alpha$ in $S\mathcal{X}$, then $\mathrm{ST\text{-}trace}(\alpha')
\;=\;
\mathrm{ST\text{-}trace}(\alpha).$
\end{proposition}
Thus symmetrisation preserves the causal pairing of events: although
$SX$ contains multiple symmetric copies of each cell, these permutations
affect only the naming of events, never the temporal structure of the
execution.
\paragraph{Pomset semantics.}
A similar compatibility holds for the ipomset interpretation of paths.
The label of a symmetric lifting of a path is obtained by permuting the ipomset label of the original path by the symmetric action.
\begin{proposition}\label{prop:glue and s commute  conseq}
    For any path $\alpha$ in a precubical set $X$ and $\beta \in SX$ lifting of $\alpha$, we have $$\sev(\alpha) \cong  \sev(\beta).$$
\end{proposition}
\begin{proof}
We employ induction on $m$ the length of $\beta$ (the same as the length of $\alpha$).
\begin{itemize}
    \item If $\beta:=(\theta.x)$ thus $\alpha:=(x)$, then $$\sev(\beta)\;=\;(\id,\ev(\theta.x),\ \id):\sev(x)  \cong \;(\id,\ \ev(x),\ \id)= \sev(\alpha).$$
\item%\textbf{(Up step / starter)} 
 If $\beta=((\theta_{j-1}.y),\ d^0_{i_j},\ (\theta_j.x))
\quad\text{thus}\quad \alpha=(y,d^0_{\theta_j^{-1}(i_j)},x))
$, then $\ev(\beta)\;:=\;\bigl(\iota_{i_j}\ \ev(\theta_j.x),\ \id \bigr).$
Further $\ev(\alpha)\;:=\;\bigl(\iota_{\theta_j^{-1}(i_j)}, \ev(x),\ \id\ \bigr):\sev(y) \to \sev(x)\cong \sev(\beta)$
\item In the case of down step, we proceed similarly.
\item If $\beta=\beta_1*\beta_2$, then $\alpha=\alpha_1*\alpha_2$ such that $\beta_{\ell}$ is a lifting of $\alpha_{\ell}$ for $\ell=\{ 0,1\}$. By induction hypothesis, there exist isomorphisms $f_1:\sev(\alpha_1)\xrightarrow{\cong}\sev(\beta_1) \text{ and } 
f_2:\sev(\alpha_2)\xrightarrow{\cong}\sev(\beta_2).$ By the equivarience condition in Lemma  \ref{lem:paths-in-SX}, and Def.\ref{def: path pomset} of $\ev$,
$f_1$ and $f_2$ coincide on the interface along which the gluing is
performed. Hence the union $f := f_1 \cup f_2$ is a well-defined bijection on the carrier of the glued pomset $\sev(\alpha_1)*\sev(\alpha_2)$.
Moreover, $f$ acts on the external interfaces. Therefore $f$ is an isomorphism i.e. $\sev(\alpha)\cong\sev(\beta)$. \qedhere
\end{itemize} 
\end{proof}
These results show that symmetrisation is compatible with both
observable semantics: ST--traces are preserved, while interval--pomset
labels are preserved up to isomorphism. 
\section{Relation between Pomset Labels and ST–Traces}\label{sec: relation}
This section makes precise the relationship between ST--traces and
interval--pomset labels.
Although both semantics are widely used to describe the behaviour of paths,
their exact correspondence depends on subtle structural conditions that are
often left implicit.
We first show how equality of ST--traces is reflected at the level of pomset
labels, and then analyse the converse direction.
\begin{lemma}\label{lemma:pomset and st trace} 
  Let $\alpha \in \Path_{\mathcal{X}}$ and $\beta \in \Path_{\mathcal{Y}}$, where $\mathcal{X}$ and $\mathcal{Y}$ are (s)HDAs. If $\alpha$ and $\beta$ have the same ST-trace, then $\sev(\alpha) \cong \sev(\beta)$.
\end{lemma}
\begin{proof}
Write $\alpha=(\alpha_0,\ldots,\alpha_n)$ and $\beta=(\beta_0,\ldots,\beta_n)$ for paths of the same length,
and let $\alpha_{\le m}$ and $\beta_{\le m}$ be their prefixes of length $m$.
We prove by induction on $m$ that $\sev(\alpha_{\le m})\cong \sev(\beta_{\le m})$ as ipomsets (Def.~\ref{def:isomorphism-ispomsets}).
\textbf{Base case $m=0$.}
Both prefixes are the initial cell. Hence $\sev(\alpha_{\le 0})$ and $\sev(\beta_{\le 0})$
are the empty ipomset, so the isomorphism is trivial.
\textbf{Induction step.}
Assume $\sev(\alpha_{\le m})\cong \sev(\beta_{\le m})$ and consider the last steps
$\alpha_{m}\to\alpha_{m+1}$ and $\beta_{m}\to\beta_{m+1}$.
Set $\pi_m:=\sev(\alpha_{\le m})=(s_m,P_m,t_m):U_m\to V_m,
$ and $
\rho_m:=\sev(\beta_{\le m})=(s'_m,Q_m,t'_m):U'_m\to V'_m.$ By the induction hypothesis, there exists an isomorphism (Def.~4.15) $(f_m,f^U_m,f^V_m):\pi_m\cong \rho_m,$ that is, $f_m:P_m\to Q_m$ is a label-preserving bijection reflecting and preserving precedence,
and $f^U_m:U_m\to U'_m$, $f^V_m:V_m\to V'_m$ are the induced interface bijections making the
interface squares commute.
$$\begin{tikzcd}[row sep=.5cm, column sep=2cm]
U_m \arrow[r, "s_m"] \arrow[d, "f^U_m" swap]
  & P_m \arrow[d, "f_m"]
  & V_m \arrow[l, "t_m" swap] \arrow[d, "f^V_m"] \\[0.2cm]
U'_m \arrow[r, "s'_m" swap]
  & Q_m
  & V'_m \arrow[l, "t'_m"]
\end{tikzcd}$$

Since ${\rm ST\text{-}trace}(\alpha)={\rm ST\text{-}trace}(\beta)$, the $(m{+}1)$-st trace symbol
for $\alpha$ and $\beta$ coincides. We distinguish two cases.

\smallskip\noindent
\emph{Case 1: the common trace symbol is $a^+$.}
Then both last steps are \emph{starts} of an $a$-labelled event.
By definition of $\sev$, we have $\sev(\alpha_{\le m+1}) \;=\; \pi_m * E^+_a,$ $
\sev(\beta_{\le m+1}) \;=\; \rho_m * {E^+_a}',$ where $E^+_a$ and ${E^+_a}'$ are the elementary ipomsets that add one fresh $a$-event,
glued along the current target interfaces $V_m$ and $V'_m$, respectively.
There is an evident isomorphism $E^+_a\cong {E^+_a}'$ whose induced bijection on the gluing
interface is precisely $f^V_m:V_m\to V'_m$ (it maps the unique new $a$-event to the unique new $a$-event).
by the stepwise definition of path labels, the carrier bijection
$f_m$ extends to an isomorphism $\sev(\alpha_{\le m+1}) \;\cong\; \sev(\beta_{\le m+1}).$
$$\begin{tikzcd}[row sep=.5cm, column sep=2cm]
P_m \arrow[r, "i_P"] \arrow[d, "f_m" swap]
  & P_{m+1} \arrow[d, dashed, "f_{m+1}"] \\
Q_m \arrow[r, "i'_Q" swap]
  & Q_{m+1}
\end{tikzcd}
\qquad
\begin{tikzcd}[row sep=.5cm, column sep=2cm]
R_m^+ \arrow[r, "i_R"] \arrow[d, "g_m" swap]
  & P_{m+1} \arrow[d, dashed, "f_{m+1}"] \\
R_m'^+ \arrow[r, "i'_R" swap]
  & Q_{m+1}
\end{tikzcd}$$

\emph{Case 2: the common trace symbol is $a^-_k$.}
Then both last steps are \emph{terminations} of the $k$-th currently-open $a$-event
(in the standard ST bookkeeping).
Let $e\in V_m$ be the interface element of $\pi_m$ corresponding to that $k$-th open $a$-event,
and let $e'\in V'_m$ be the corresponding element of $\rho_m$.
Because the prefixes have the same ST-trace and $(f_m,f^U_m,f^V_m)$ is an ipomset isomorphism,
the induced bijection $f^V_m$ respects this bookkeeping, hence $f^V_m(e)=e'$. By definition of $\sev$, we have $\sev(\alpha_{\le m+1}) \;=\; \pi_m * E^-_{a,e},
$ $
\sev(\beta_{\le m+1}) \;=\; \rho_m * {E^-_{a,e'}}',$ where $E^-_{a,e}$ (resp.\ ${E^-_{a,e'}}'$) is the elementary ipomset that terminates the distinguished
interface element $e$ (resp.\ $e'$), glued along $V_m$ (resp.\ $V'_m$).
There is an isomorphism $E^-_{a,e}\cong {E^-_{a,e'}}'$ whose induced bijection on the gluing interface
is $f^V_m$ and which maps the distinguished element $e$ to $e'$.
$$\begin{tikzcd}[row sep=.5cm, column sep=2cm]
V_m \arrow[r, "s^-"] \arrow[d, "f^V_m" swap]
  & R_m^- \arrow[d, "g_m"] \\
V'_m \arrow[r, "s'^-" swap]
  & R_m'^-
\end{tikzcd}$$
By the stepwise definition of path labels, we obtain $\sev(\alpha_{\le m+1}) \;\cong\; \sev(\beta_{\le m+1}).$ Thus in either case the isomorphism extends from length $m$ to length $m{+}1$.
By induction, $\sev(\alpha)\cong \sev(\beta)$.
\end{proof}

One might expect the converse to hold, namely that
$\sev(\alpha)\cong \sev(\beta)$ implies equality of ST--traces.
This is not true in general.
Ipomset semantics abstract away from the precise temporal order of starts
and terminations of concurrent events and therefore identify paths that are merely congruent.
For instance, the ST--traces
$a^{+} b^{+} a^{1} b^{2}$ and $b^{+} a^{+} a^{2} b^{1}$
are distinct but induce the same interval ipomset
$(a \parallel b)$. Nevertheless, ipomset isomorphism retains enough information to recover
temporal behaviour up to congruence.
Achieving this requires additional structure, which we develop next.
\begin{definition}
    Let $\alpha=(x_{0}, d_{i_1}^{k_1}, x_{1}, d_{i_2}^{k_2}, \ldots, d_{i_m}^{k_m}, x_{m})$ and $\beta=(y_{0}, d_{r_1}^{k_1}, y_{1}, d_{r_2}^{k_2}, \ldots, d_{r_m}^{k_m}, y_{m})$. We say that $\alpha$ and $\beta$ have matching events if $\ev(x_{j-1},d^{k_j}_{i_j}, x_j)=\ev(y_{j-1},d^{k_j}_{r_j},y_j)$. We write $\alpha\equiv \beta$.
\end{definition}
\begin{lemma}\label{lemma: new lemma3 ev then ST}
    Let $\mathcal{X}$ and $\mathcal{ Y}$ be (s)HDAs, and let
$\alpha\in\Path_{\mathcal{ X}}$ and $\beta\in\Path_{\mathcal{ Y}}$. If $\alpha \equiv \beta$, then ST-trace$(\alpha)=$ST-trace$(\beta)$.
\end{lemma}
\begin{proof}
By definition of $\equiv$, the paths $\alpha$ and $\beta$ have the same sequence of face maps.
That is, we can write $\alpha = (x_0, d_{i_1}^{k_1}, x_1, \dots, d_{i_{m+1}}^{k_{m+1}},$ $ x_{m+1}),
$ and $
\beta = (y_0, d_{i_1}^{k_1}, y_1, \dots, d_{i_{m+1}}^{k_{m+1}}, y_{m+1}),$ for the same indices $(i_j,k_j)$. Since the ST-trace construction depends only on the sequence of face maps
$(d_{i_1}^{k_1},\dots,d_{i_{m+1}}^{k_{m+1}})$ and not on the intermediate cells,
it follows immediately that ${\rm ST\text{-}trace}(\alpha)={\rm ST\text{-}trace}(\beta)$.
\end{proof}
The notion of matching events captures exactly the stepwise information
needed to determine the ST--trace. However, matching events alone is insufficient to relate arbitrary paths.
To realise ipomset equivalence stepwise, we therefore pass to symmetric
HDAs, where all permutations of concurrent events are explicit. 
\begin{lemma}\label{lemma: Sev and SP}
Let $X$ be a precubical set, and let $\alpha\in\Path_X$ and
$P$ an ipomset. If $\sev(\alpha)\cong P$, then there exists $\alpha'\in\Path_{SX}$ lifting of $\alpha$ such that $\sev(\alpha')=P$
\end{lemma}
\begin{proof}
    Fix $\alpha'\in \Path_{SX}$ lifting of $\alpha$. We employ induction on $m$, the length of $\alpha$.
    \begin{itemize}
        \item If $\alpha=(x)$ there exists a unique permutation $\sigma$ such that $\sigma(\sev(x))=P$.
Define $\alpha':=(\sigma.x)$. Hence $\ev(\alpha') = P$, establishing the base case.
        \item If $\alpha=\bigl(y,\,d^0_i,\,x\bigr)$. On the one hand
$\sev(\alpha)=\bigl(\iota_i,\ \sev(x),\id\bigr): \sev(y) \to \sev(x).$
%\]
Since $\sev(\alpha)\cong P$ by an isomorphism $f$ (defined uniquely by permutation $\theta \in \SG$), we have the following diagram
$$\begin{tikzcd}[row sep=.5cm, column sep=2cm]
U_m \arrow[r, "\iota_i"] \arrow[d, "f^U" swap]
  & \sev(x) \arrow[d, "f"]
  & \sev(x) \arrow[l, "\id" swap] \arrow[d, "f"] \\[0.2cm]
U \arrow[r, "\iota_{f(i)}" swap]
  & V
  & V \arrow[l, "\id"]
\end{tikzcd}$$
Define $\alpha^{\prime}=\bigl((d_i\theta. y),d^0_{\theta(i)},(\theta.x)\bigl)$. By definition, $\sev(\alpha)=P$.
\item If $\alpha=\bigl(x,\,d^1_i,\,x\bigr)$, we proceed similarly by taking $\alpha^{\prime}=\bigl((\theta.x),d^0_{\theta(i)},(d_i\theta. y)\bigl)$
\item If $\alpha=\alpha_{1}*\alpha_{2}$ is a concatenation where $\alpha_1$ and $\alpha_2$ are shorter than $\alpha$ and $P\cong \ev(\alpha)$ via an isomorphism $f$, write $P_i$ for $i=1,2$ for the ipomsets such that $P_i=f(\alpha_i)$. By the induction hypothesis, there exist $\alpha_1^{\prime}$ and $\alpha_2^{\prime}$ lifting of $\alpha_1$ and $\alpha_2$ such that $\ev(\alpha_i^{\prime})=P_i$. Take $\alpha^{\prime}=\alpha_{1}^{\prime}*\alpha_{2}^{\prime}$, so that $\sev(\alpha^{\prime})=\sev(\alpha_{1}^{\prime}*\alpha_{2}^{\prime})=P_1*P_2=P$.
   \end{itemize}
\end{proof}
\begin{lemma}\label{lemma: Sev and Sev}
Let $X$ and $Y$ be precubical sets, and let $\alpha\in\Path_X$ and
$\beta\in\Path_Y$. If $\sev(\alpha)\cong \sev(\beta)$, then there exist $\alpha'\in\Path_{SX}$ lifting of $\alpha$ and $\beta'\in\Path_{SY}$ lifting of $\beta$ such that $\ev(\alpha')=\sev(\beta')$.
\end{lemma}
\begin{proof}
    Apply Lemma \ref{lemma: Sev and SP} for $P=\ev(\beta)$.
\end{proof}
Now we align paths up to congruence while preserving matching events.
\begin{lemma}\label{lemma: from ev = to ev equiv}
    Let $X$ and $Y$ be (s)precubical sets, and let $\alpha\in\Path_X$ and
$\beta\in\Path_Y$. If $\ev(\alpha)=\ev(\beta)$ then there exists a path $\gamma\simeq \alpha \in\Path_X$ such that $\gamma \equiv \beta$.
\end{lemma}
\begin{proof}
    We employ induction on the length of $\alpha$. %Recall that the ST--trace is defined only for paths in HDA, that is starting from the iniial cell.
    \begin{itemize}
        %\item If $\alpha=(x)$, then just take $\gamma=\alpha$.
        \item If $\alpha=(x_1,d_i^k,x_2,d_j^k,x_3)$ for $k=0,1$,  then since $\ev(\alpha)=\ev(\beta)$, $\beta=(y_1,d_r^k,y_2,d_l^k,y_3)$ such that $\ev(x_s)=\ev(y_s)$ for $s=1,3$. Thus, there are two cases:
        \begin{itemize}
            \item if $i=r$ and $j=l$ then take $\gamma= \alpha$.
            \item If $i=l$ and $j=r$ then take the replacement segment (exists and unique as detailed in Definition \ref{def:adjacency}) $\gamma= (y_1,d_l^k,y^{\prime}_2,d_r^k,y_3)$.
        \end{itemize}
        \item If $\alpha=\alpha_{1}*\alpha_{2}*\alpha_3$ and $\beta=\alpha_{1}*\beta_{2}*\alpha_3$ such that $\alpha_2\simeq \beta_2$. The required condition follows immediately from the induction hypothesis.\qedhere
    \end{itemize}
\end{proof}
Combining the previous constructions, ipomset equivalence can be lifted to
stepwise agreement.%, thus equality of ST--traces (Lem. \ref{lemma: new lemma3 ev then ST}), in the symmetric setting.
\begin{lemma}\label{lemma:liftings}
    Let $\mathcal{X}$ and $\mathcal{Y}$ be HDAs, and let $\alpha\in\Path_{\mathcal{X}}$ and $\beta\in\Path_{\mathcal{Y}}$. If $\sev(\alpha)\cong \sev(\beta)$, then there exist $\gamma\in\Path_{S\mathcal{X}}$ lifting of $\alpha$, $\beta'\in\Path_{S\mathcal{Y}}$ lifting of $\beta$, and $\alpha' \cong \gamma $ such that ST-trace$(\alpha')=$ST-trace$(\beta')$.
\end{lemma}
\begin{proof}
    Let $\beta'=S\beta$. By Lemma \ref{lemma: Sev and Sev}, there exists $\gamma\in\Path_{S\mathcal{X}}$ lifting of $\alpha$ such that $\ev(\gamma)=\sev(\beta)$. By Lemma \ref{lemma: from ev = to ev equiv}, there exists $\alpha' \simeq \gamma$ such that $\alpha'\equiv \beta'$. By Lemma \ref{lemma: new lemma3 ev then ST}, ST--trace$(\alpha')=$ST--trace$(\beta')$.
\end{proof}
We can now return to the original HDAs.
\begin{proposition}\label{prop: from pomset to st}
    Let $\mathcal{X}$ and $\mathcal{Y}$ be HDAs, and let $\alpha\in\Path_{\mathcal{X}}$ and $\beta\in\Path_{\mathcal{Y}}$. If $\sev(\alpha)\cong \sev(\beta)$, then there exists $\gamma\simeq \alpha$  such that ST-trace$(\gamma)=$ST-trace$(\beta)$.
\end{proposition}
\begin{proof}
    By Lemma\ref{lemma:liftings}, there exist $\alpha^{\prime}\in\Path_{S\mathcal{X}}$ lifting of $\alpha$, $\beta'\in\Path_{S\mathcal{X}}$ lifting of $\beta$, and $\gamma'\cong \alpha' $ such that ST-trace$(\gamma')=$ST-trace$(\beta')$. Further, by {\cite[Prop.~5.2]{KAHL202247}}, there exists $\gamma \in \Path_{X}$ such that $\gamma \simeq\alpha$ and $\gamma^{\prime}$ lifting of $\gamma$. On the one hand we have, ST-trace$(\gamma')=$ST-trace$(\beta')$. On the other hand, by Proposition \ref{lem:ST-invariant-under-S}, ST--trace$(\gamma')=$ST--trace$(\gamma)$ and ST--trace$(\beta)=$ST--trace$(\beta')$. Thus, ST--trace$(\beta)=$ST--trace$(\gamma)$.
\end{proof}
\section{Bisimulations for Higher Dimensional Automata}\label{sec: Bis for HDA}
We reformulate history preserving and hereditary history preserving
bisimulation for HDAs inipomset terms, connecting them to open maps
\cite{JOYAL1996164} and language semantics \cite{LanguageofHDA}.
 \begin{definition} %\label{def: hp bisimulation}
 A \emph{history preserving bisimulation} (hp-bisimulation) between HDAs $\mathcal{Y}$ and $\mathcal{Z}$ is a symmetric relation $R$ between paths in $Y$ and $Z$ such that
\begin{enumerate}
\item \label{enn: h bisim.e1}initial paths $(i_Y)$ and $(i_Z)$ are related;
\item \label{enn: h bisim.e2} for all $( \rho, \sigma)\in R$, % $sev(\alpha) \cong sev(\beta)$%\safareplace{ $ST$-$trace(\rho)=ST$-$trace(\sigma)$}{$sev(\alpha) \cong sev(\beta)$};
ST--trace$(\rho)=ST$-$trace(\sigma)$; 
\item \label{enn: h bisim.e3} for all $( \rho, \sigma)\in R$ and path $\rho'$ in $Y$ such that $\rho$ and $\rho'$ may be concatenated,
  there exists a path $\sigma'$ in $Z$ such that $( \rho * \rho',
  \sigma* \sigma')\in R$;
\item \label{enn: h bisim.e4} for all $(\rho, \sigma)\in R$ and path $\rho'$ in $Y$ such that $\rho \stackrel{\ell}{\leftrightsquigarrow} \rho'$, there exists a path $\sigma'$ in $Z$ such that $\sigma \stackrel{\ell}{\leftrightsquigarrow}  \sigma'$ and $(\rho',\sigma')\in R$;
\end{enumerate}
  The relation $R$ is called \emph{hereditary history preserving bisimulation} (hhp-bisimulation) if, in addition, it satisfies:
\begin{enumerate}
\setcounter{enumi}4
\item \label{enn: h bisim.e6} for all $( \rho, \sigma)\in R$ and $\rho'$ restriction of $\rho$, there exists $\sigma'$ restriction of $\sigma$ such that $( \rho',\sigma')\in R$.
\end{enumerate}
We say that $\mathcal{X}$ and $\mathcal{Y}$ are (hereditary) \emph{history preserving bisimilar} and write $\mathcal{X}  \approx_{(h)hp} \mathcal{Y}$ if there exists a (hereditary) \emph{history- preserving bisimulation} $R$ between them; this is an equivalence relation.
\end{definition}
 ST-bisimulation is defined as the hp-bisimulation but dropping clause 4 and 5.
\begin{theorem}\label{th:st-bis main result}
    Two HDAs $\mathcal{X}$ and $\mathcal{Y}$ are ST-bisimilar \emph{iff}
there exists
\(
   R\subseteq\!\Path_{\mathcal{X}}\times\Path_{\mathcal{Y}}
\)
such that:
\begin{enumerate}
\item initial paths $(i_Y)$ and $(i_Z)$ are related;
\item  $(\alpha,\beta)\in R
   \;\Longrightarrow\;
   \sev(\alpha)\cong \sev(\beta);$
\item $R$ respects path initial inclusion: for all $( \rho, \sigma)\in R$ and path $\rho'$ in $Y$ such that $\rho$ and $\rho'$ may be concatenated,
  there exists a path $\sigma'$ in $Z$ such that $( \rho * \rho',
  \sigma* \sigma')\in R$;
  \end{enumerate}
\end{theorem}
\begin{proof}
$"\Rightarrow" $ Let $R$ be an ST--bisimulation between $\mathcal{X}$ and $\mathcal{Y}$. By Lemma~\ref{lemma:pomset and st trace}, $R$ satisfies the three conditions
above.

\emph{``$\Leftarrow$''} Let $K$ be a relation satisfying the stated conditions.
By Proposition~\ref{prop: from pomset to st}, $K$ induces a relation $R \;=\;
\{\,(\alpha',\beta) \mid (\alpha,\beta)\in K,\;
\alpha'\simeq\alpha,\;
\text{ST--trace}(\alpha')=\text{ST--trace}(\beta)\,\}.$

By construction, $R$ is an ST--bisimulation relating $\mathcal{X}$ and $\mathcal{Y}$.
\end{proof}
\begin{theorem}\label{th: hhp-bis main result}
Two HDAs $ \mathcal{X}$ and $\mathcal{Y}$ are \textup{(h)hp}–bisimilar \emph{iff}
there exists a relation
\(
   R\subseteq\!\Path_{\mathcal{X}}\times\Path_{\mathcal{Y}}
\)
that satisfies conditions \textup{1, 3, 4, (5)}
of the hhp–bisimulation definition,
\emph{and} the following replacement of clause~2: $(\alpha,\beta)\in R
   \;\Longrightarrow\;
   \sev(\alpha)\cong \sev(\beta).$
\end{theorem}
\begin{proof}
$(\Rightarrow)$
Let $R$ be an (h)hp--bisimulation between $\mathcal{X}$ and $\mathcal{Y}$.
Arguing as in the proof of Th~\ref{th:st-bis main result}, $R$ satisfies
the stated conditions, with ipomset isomorphism replacing equality of
ST--traces.

$(\Leftarrow)$
Let $K$ be a relation satisfying the stated conditions.
Construct the relation $R$ as in the proof of
Th~\ref{th:st-bis main result}.
%By Lemma~\ref{lem: cong and adj}, 
For clause~4 of hp--bisimulation, Since $(\rho,\sigma)\in R$, there exists $\alpha$ such that
$\rho\simeq\alpha$, $(\alpha,\sigma)\in K$, and
$\mathrm{ST\text{-}trace}(\rho)=\mathrm{ST\text{-}trace}(\sigma)$.
By Definition~\ref{def:path-congruence}, the congruence $\rho\simeq\alpha$
is generated exclusively by the reversible rules~(1) and~(2) of
Definition~\ref{def:adjacency}.

\medskip
\noindent\textbf{Step 1: $\sigma$ is $\ell$-swappable.} Since $\mathrm{ST\text{-}trace}(\rho)=\mathrm{ST\text{-}trace}(\sigma)$,
the polarity of the face maps at positions $\ell$ and $\ell+1$
coincides in $\rho$ and $\sigma$. \emph{Rules~(1) and~(2).} The ST-trace equality forces $\sigma$
  to have two consecutive face maps of the same polarity at positions
  $\ell$ and $\ell+1$. Since rules~(1) and~(2) are reversible and
  apply to any two consecutive same-polarity steps regardless of index
  ordering, $\sigma$ is $\ell$-swappable, yielding $\sigma
  \stackrel{\ell}{\leftrightsquigarrow}\sigma'$. \emph{Rules~(3) and~(4).} The ST-trace records
  $\mathrm{start}(i_{\ell+1})$ explicitly at every down-step.
  Since $\mathrm{ST\text{-}trace}(\rho)=\mathrm{ST\text{-}trace}(\sigma)$,
  we have $\mathrm{start}_\rho(i_{\ell+1})=\mathrm{start}_\sigma(r_{\ell+1})$.
  Hence the directed swap condition holds for $\sigma$ at position
  $\ell$, giving $\sigma\stackrel{\ell}{\leftrightsquigarrow}\sigma'$. \noindent\textbf{Step 2: $(\rho',\sigma')\in R$.} We must find $\alpha'$ such that $\rho'\simeq\alpha'$,
$(\alpha',\sigma')\in K$, and
$\mathrm{ST\text{-}trace}(\rho')=\mathrm{ST\text{-}trace}(\sigma')$.
We proceed by induction on the length $n$ of the congruence chain
from $\rho$ to $\alpha$. \emph{Base case $n=0$, i.e.\ $\rho=\alpha$.}
Then $\alpha\stackrel{\ell}{\leftrightsquigarrow}\rho'$ is the
same swap as on $\rho$. Clause~(4) of $K$ applied to
$(\alpha,\sigma)\in K$ yields $\sigma\stackrel{\ell}
{\leftrightsquigarrow}\sigma'$ and $(\alpha',\sigma')\in K$ where
$\alpha'=\rho'$. Hence $\rho'\simeq\alpha'$ trivially and
$\mathrm{ST\text{-}trace}(\rho')=\mathrm{ST\text{-}trace}(\sigma')$
by the ST-trace analysis of Step~1. Thus $(\rho',\sigma')\in R$. \emph{Inductive step $n>0$.}
By Definition~\ref{def:path-congruence}, we can write  $ \rho = \gamma_1 * \rho_0 * \delta$ and $ \alpha = \gamma_1 * \alpha_0 * \delta,$ where $\rho_0\simeq\alpha_0$ via a congruence chain of length
$n-1$, using the third clause of Definition~\ref{def:path-congruence}.
We distinguish three sub-cases according to the position of the
$\ell$-swap.

\begin{itemize}
  \item \emph{Swap inside $\gamma_1$ or inside $\delta$.} The
  sub-segment $\rho_0$ is unaffected. The result $\rho'$ has the
  same decomposition with $\rho_0$ unchanged, so $\rho'\simeq\alpha'$
  where $\alpha'$ is obtained by applying the same swap to $\alpha$,
  and $(\alpha',\sigma')\in K$ by clause~(4) of $K$.

  \item \emph{Swap entirely inside $\rho_0$.} The induction
  hypothesis applies directly to $\rho_0$ and $\alpha_0$, yielding
  $\alpha_0'$ with $\rho_0'\simeq\alpha_0'$ and
  $(\gamma_1*\alpha_0'*\delta,\sigma')\in K$. Setting
  $\alpha':=\gamma_1*\alpha_0'*\delta$ gives $\rho'\simeq\alpha'$.

  \item \emph{Swap at the boundary between $\rho_0$ and $\delta$
  (or between $\gamma_1$ and $\rho_0$).} Absorb the two boundary
  steps into an extended segment
  $\rho_0^+ = \rho_0 * (\text{first step of }\delta)$
  and $\delta^- = \delta$ minus its first step, so that
  $\rho = \gamma_1 * \rho_0^+ * \delta^-$ and the swap at $\ell$
  now falls entirely inside $\rho_0^+$. The induction hypothesis
  applies to $\rho_0^+$ and $\alpha_0^+$ (defined analogously),
  reducing to the previous sub-case.
\end{itemize}

In all cases we obtain $\alpha'$ with $\rho'\simeq\alpha'$,
$(\alpha',\sigma')\in K$, and
$\mathrm{ST\text{-}trace}(\rho')=\mathrm{ST\text{-}trace}(\sigma')$,
so $(\rho',\sigma')\in R$. In the hereditary case, clause~5 follows directly by construction.
Hence $R$ is an (h)hp--bisimulation.
\end{proof}
\section{Conclusion}

We have developed an order-free semantic foundation for 
higher-dimensional automata. The central technical 
contributions are the categorical isomorphism between 
HDAs over the unordered base $\Xi$ and symmetric HDAs, 
the canonical assignment of interval ipomsets to execution 
paths, the formal correspondence between ipomset labels 
and ST-traces, and the characterization of ST- and 
hhp-bisimulation via ipomset isomorphism. These results 
have concrete consequences beyond what is explicitly 
developed here.

On the logical side, the order-free path-category 
structure established here provides the missing foundation 
for deriving, via the Open Maps framework, modal logics 
that canonically characterize hhp-bisimulation; the 
temporal and modal logics of 
\cite{AMRANE2025115156, amrane2025buchielgottrakhtenbrottheoremhigherdimensionalautomata, Safa} 
are natural targets for revisiting under this symmetric 
foundation. On the structural side, the elimination of 
the event order artifact opens a cleaner path toward 
systematic translations between HDAs and other models 
of concurrency such as Petri nets, where the mismatches 
documented in \cite{ulipetri, efficintconversion} 
were a direct consequence of the representational 
incompatibility resolved here.
% Manual bibliography; force numeric citations by ignoring optional author--year labels in \bibitem
\makeatletter
\let\latex@bibitem\bibitem
\renewcommand\bibitem[2][]{\latex@bibitem{#2}}
\makeatother

\bibliographystyle{ACM-Reference-Format}
%\bibliography{sample-base}

%\newpage
%\appendix

\end{document}